\documentclass[12pt]{article}
\pdfoutput=1
\usepackage{amsfonts,amsmath,amssymb,amsthm,graphicx}
\usepackage{enumerate,enumitem}
\usepackage{setspace} % for spacing
\usepackage{hyperref}
\usepackage[capitalize]{cleveref}
\usepackage{ifthen}
\usepackage[margin=1.25in]{geometry}
\usepackage[font=footnotesize,labelfont=bf]{caption}

\usepackage{fix-cm}
\usepackage{graphicx}
\usepackage[titletoc]{appendix}
\usepackage{xcolor}
\usepackage{amsmath}
\usepackage{amsthm}
\usepackage{xifthen}
\usepackage{stackengine}
\usepackage{comment}
\usepackage[utf8]{inputenc}
\usepackage{nicefrac}  
\usepackage{xfrac}
\usepackage{enumerate}
\usepackage{CJKutf8} 
\usepackage{cleveref} 
\usepackage{subcaption} % 核心包：用于 subfigure
\usepackage{bbm}
\usepackage{natbib}

\usepackage{enumitem}
\theoremstyle{plain}

\newtheorem{theorem}{Theorem}

\newtheorem{corollary}{Corollary}

\newtheorem{definition}{Definition}
\newtheorem{example}{Example}

\newtheorem{lemma}{Lemma}

\newtheorem{proposition}{Proposition}

\newcommand{\squishlist}{
\begin{list}{{{\small{$\bullet$}}}}
{\setlength{\itemsep}{3pt}      \setlength{\parsep}{1pt}
\setlength{\topsep}{1pt}       \setlength{\partopsep}{0pt}
\setlength{\leftmargin}{1em} \setlength{\labelwidth}{1em}
\setlength{\labelsep}{0.5em} } }
\newcommand{\squishend}{  \end{list}}

\newcommand{\omt}[1]{} % Comment out content

\usepackage{tikz}
\usetikzlibrary{arrows.meta}
\usetikzlibrary{patterns, decorations.pathmorphing, decorations.pathreplacing, plotmarks}

\usepackage{pgfplots}
\pgfplotsset{compat=newest}
\usepgfplotslibrary{groupplots}

\usepackage{color-edits}
\addauthor[Boli]{xbl}{blue}
\addauthor[Zhicheng]{dzc}{red}
\addauthor[Yingkai]{yl}{orange}

\newcommand{\cc}[1]{\ensuremath{\mathsf{#1}}}

\newcommand{\prob}[2][]{\text{Pr}\ifthenelse{\not\equal{}{#1}}{_{#1}}{}\!\left[{\def\givenn{\middle|}#2}\right]}
\newcommand{\expect}[2][]{\mathbb{E}\ifthenelse{\not\equal{}{#1}}{_{#1}}{}\!\left[{\def\givenn{\middle|}#2}\right]}
\newcommand{\indicator}[2][]{\mathbf{1}\ifthenelse{\not\equal{}{#1}}{_{#1}}{}\!\left\{{\def\givenn{\middle|}#2}\right\}}
\newcommand{\Variance}[2][]{\text{Var}\ifthenelse{\not\equal{}{#1}}{_{#1}}{}\!\left[{\def\givenn{\middle|}#2}\right]}

\newcommand{\E}{\mathbb{E}}
\newcommand{\ind}[1]{\mathbf{1}{\left\{#1\right\}}} % Simplified version

\newcommand{\dd}{\ \mathrm{d}}

\newcommand{\testset}{\mathcal{T}}

\newcommand{\signalSpace}{\mathcal{S}}
\newcommand{\statespace}{\Theta}
\newcommand{\messageSpace}{\mathcal{M}}
\newcommand{\outcome}{\omega}
\newcommand{\outcomeSpace}{\Omega}

\newcommand{\supp}{\cc{supp}}

\newcommand{\prior}{F}

\newcommand{\MPC}{\cc{MPC}}

\title{\Large Going Public: Communication in Collective Decisions\thanks{We thank Xiao Lin, Rakesh Vohra and audiences at SUFE and SJTU for helpful comments and suggestions. Yingkai Li also thanks NUS Start-up grant for financial support.}}

\author{Zhicheng Du\thanks{School of Artificial Intelligence, Renmin University of China.
Email: \texttt{duzhicheng@ruc.edu.cn}}
\and Yingkai Li\thanks{Department of Economics, National University of Singapore.
Email: \texttt{yk.li@nus.edu.sg}} 
\and Boli Xu\thanks{Department of Economics, University of Iowa.
Email: \texttt{boli-xu@uiowa.edu}}} %Tippie College of Business, 

\begin{document}

\date{}
\maketitle

\begin{abstract}

\noindent A principal and $n\ge 2$ agents can launch a project if the principal proposes it and at least $k$ agents accept. Their individual payoffs from the project depend on an ex ante unknown state. The principal can conduct a test to learn about the state and then communicate her findings to the agents via cheap talk. This paper focuses on comparing two communication regimes: public and private messaging. We show that public messaging is weakly dominant: any outcome implementable under private messaging can also be implemented under public messaging. Moreover, in a canonical environment with linear payoffs, we characterize the principal's optimal test in each regime and show that public messaging can be strictly dominant if and only if there exist two agents who are the principal's \textit{conflicting allies}.  \\

\vspace{6pt}
\noindent \textbf{Keywords:} public vs. private communication, learning, cheap talk, $k$-majority voting. 

\vspace{3pt}

\noindent \textbf{JEL Codes:} D70, D82, D83. 
\end{abstract}

\thispagestyle{empty}
\newpage
\setcounter{page}{1}

\section{Introduction}
\label{sec:intro}

Many consequential decisions are made by collective bodies—--boards, committees, councils, panels, or juries—--whose members often have heterogeneous interests. In many such settings, the decision process is led by a chairperson who typically bears two key responsibilities. First, she determines which information is generated to guide the decision: which investigation to conduct, which reference to consult, or which data to collect. Second, she can communicate her findings to other members to build support for the decision. 

In many such situations, while the information-gathering process is observable, the chairperson’s communication takes the form of soft messages, such as subjective evaluations or impressions, confidential references, or recommendations that summarize complex evidence. For example, when a corporate board votes on whether to pursue a business plan, the chairperson may conduct due diligence, yet ultimately provide a subjective evaluation; when a promotion committee evaluates a candidate, the chairperson may rely on subjective teaching observations and confidential reference checks; when a government council decides on a procurement plan, the chairperson may gather information about alternative plans and make a final recommendation. A common challenge across these settings is that the chairperson cannot commit to truthfully reporting what she has learned. 

This paper studies the interaction between endogenous learning and non-committal communication in collective decision-making, focusing on a salient design choice: should communication be public or private? Public communication includes open meetings, circulated reports, and common briefings, while private communication includes one-on-one conversations, back-channel coordination, and individualized messages. Given the chairperson's \textit{limited commitment} in communication, as described above, she faces a central trade-off between \textit{flexibility} and \textit{credibility} when choosing the communication regime. While private communication affords greater flexibility for targeted persuasion, public communication generates common knowledge among all members, which can be essential for sustaining credibility. 

The paper delivers two main findings. First, the flexibility afforded by private communication is valueless, making public communication weakly dominant. Second, in situations where two agents are the principal's \textit{conflicting allies}, the credibility offered by public communication is valuable, rendering public communication strictly dominant. 

The paper builds a model in which a principal and $n\ge 2$ agents face a collective decision of whether to launch a project. The project is launched if and only if the principal proposes it and at least $k$ agents approve of the proposal. Players' individual payoffs from the project are determined by a state drawn from a continuous distribution. The game begins with the principal publicly adopting a test to learn about the payoff-relevant state. After privately observing the test result, the principal can propose the project and communicate her findings to each agent via cheap talk. The agents then simultaneously decide whether to accept or reject the proposal. We focus on comparing two communication regimes: \textit{public messaging} requires that all agents always receive the same message, while private messaging does not impose such a constraint.  

\Cref{sec:general} presents the first main finding. \Cref{thm:public-implements-more} shows that, for a given test, any implementable outcome under private messaging is also implementable under public messaging. The intuition is as follows. Under private messaging, conditional on being willing to propose the project, the principal has a transparent incentive to always send each agent the message that maximizes his probability of acceptance (\Cref{prop:pbe_characterization}). As a result, beyond the proposal itself, the principal's message conveys no additional payoff-relevant information to individual agents, and therefore, the flexibility afforded by private messaging has no value (\Cref{prop:private-babbling}). An immediate implication of \Cref{thm:public-implements-more} is that, taking into account the principal's optimally chosen test, she also (weakly) prefers public messaging to private messaging. 

This result naturally raises the question of when public messaging can be \textit{strictly} dominant. \Cref{sec:linear} addresses this question in a canonical environment where players' payoffs are linear in the state. We characterize the optimal test under each communication regime (Propositions~\ref{prop:publiclinear} and \ref{prop:privatelinear}). Based on these characterizations, \Cref{thm:strict} shows that public messaging can be strictly dominant if and only if there exist two agents who are the principal's \textit{conflicting allies}---that is, they each share some common interests with the principal, yet are in conflict with each other. Intuitively, the credibility afforded by public communication is valuable in such situations: by publicly siding with one of the two conflicting allies, which is selected at random, the principal can secure the trust of the side she chooses.

\subsection{Related Literature}
\label{sec:literature}

First, this paper contributes to the literature on information design in voting. The closest related work is \citet{chan2019pivotal}, which studies how a sender can optimally use private communication to persuade multiple receivers under majority rule.\footnote{Relatedly, \citet{alonso2016persuading} analyzes public persuasion under majority rule, while \citet{bardhi2018modes} studies persuasion under unanimity rule. See also \citet{xu2023}.} A key distinction between the two settings concerns the sender's commitment in communication. Whereas the sender has full commitment in \citet{chan2019pivotal}, she is subject to limited commitment in our model. This difference leads to a sharp contrast in results. While \citet{chan2019pivotal} shows that private communication is better under full commitment, we find the opposite under limited commitment. 

Second, this paper speaks to the literature on information design with limited commitment. In our setting, the principal can commit to the information structure, but cannot commit to truthful reporting of the realized signal.\footnote{Some related information design models allow for misreporting but impose additional constraints. For instance, in \citet{guo2021costly} and \citet{nguyen2021bayesian}, misreporting is costly, with the cost increasing in the distance between the reported and actual signals. In \citet{lin2024credible}, misreporting is permitted provided that the distribution of the reported signals matches that of the actual signals.} This feature also appears in \citet{lipnowski2022persuasion, kreutzkamp2024persuasion, lyu2022information, li2024information}.\footnote{The companion paper, \citet{li2024information}, analyzes how the principal's limited commitment leads to simple and realistic predictions of the optimal test in settings where the agent's type space and action space are complex. This paper, by contrast, focuses on comparing different communication regimes in settings with multiple agents.} We differ from these papers by focusing on explicitly comparing different communication regimes in a multiple-agent setting. 

Third, this paper speaks to the literature on cheap talk with multiple receivers. The seminal work in this strand of literature is \citet{farrell1989cheap}, which highlights the flexibility-credibility trade-off. Building on this insight, \citet{goltsman2011talk} explicitly compares public and private communication. Despite the thematic similarity, our setting differs substantially from these two papers, as we focus on the context of $k$-majority voting.\footnote{Some related papers study cheap talk in voting but do not compare the communication regimes. For example, \citet{schnakenberg2015expert} studies only public communication, while \citet{schnakenberg2017informational, salcedo2019persuading} analyze only private communication. These papers also take the information structure as exogenous.} In our setting, upon proposing the project, the principal has a transparent incentive to maximize voter support,\footnote{Our notion of \textit{transparent incentive} is related to, but distinct from, the notion of \textit{transparent motive} in \citet{lipnowski2020cheap}. In their setting, the sender's payoff is state-independent, while in ours, the sender's payoff is state-dependent, although her ordinal preference over a receiver's actions is state-independent conditional on proposing the project.} which renders the flexibility afforded by private communication completely valueless and makes public communication weakly dominant. This contrasts with \citet{farrell1989cheap} and \citet{goltsman2011talk}, which focus on decision environments that are independent across audiences and find that either public or private communication may be preferable. Beyond differences in setting, our paper also endogenizes the information structure, yielding additional insights into the sender's optimal learning decisions.

Finally, this paper adds to the study of dispersed information in collective decisions. \citet{ali2025political} shows that informational dispersion can lead to a ``deadlock'' in collective decisions with exogenous information environments. In contrast, our paper examines how a chairperson can design the information environment to overcome such a deadlock, even when she has limited commitment in communication.

\section{Model}
\label{sec:prelim}

\subsection{Setup}
\label{subsec:setup}

A principal (``she'') and $n\ge 2$ agents (each ``he'', indexed by $i\in[n]$) decide whether to launch a project. The project is launched if and only if the principal proposes it and at least $k\le n$ agents accept the proposal. Players' individual payoffs from the project are governed by a state $\theta \in \statespace:= [-1, 1] $, which is drawn from a distribution represented by its cumulative distribution function $\prior$. We assume that $\prior$ is continuous, meaning that the distribution of $\theta$ is atomless, which also allows us to let $f$ denote the corresponding probability density function.

If the project is not launched, all the players get a payoff normalized to zero.
If the project is launched, the principal gets a payoff of $u(\theta)$, and each agent $i$ gets $v_i(\theta)$. 
For technical reasons, we let the players' payoff functions be bounded. 
To make the analysis non-trivial, we also assume that $\inf_\theta [u(\theta)] < 0 < \sup_\theta [u(\theta)]$ and $\inf_\theta [v_i(\theta)] < 0 < \sup_\theta [v_i(\theta)]$ for any $i$; that is, every player may benefit or suffer from launching the project. All the players are risk-neutral. 

None of the players observes $\theta$ directly. However, the principal can conduct a \textit{test} to learn about $\theta$ at no cost. A test $t=(\signalSpace,\pi)$ is described by (a) a measurable signal space $\signalSpace$ and (b) a signal generation rule $\pi$ that maps the state into a distribution of the signals. We let $\testset$ denote the collection of all possible tests.

The game proceeds as follows. 
\begin{itemize}%[(i)]
    \item \textbf{Learning stage.}~
    The state $\theta$ is realized by $F$. The principal publicly chooses a test $t=(\signalSpace,\pi)$ and privately observes the test result $s\sim \pi(\cdot\mid \theta)$. 
    \item \textbf{Proposing stage.}~
    The principal decides whether to propose or forgo the project. If she forgoes the project ($P=0$), the game ends. If she proposes the project ($P=1$), she pays a fixed cost of $c>0$, and the game proceeds to the next stage.  
    \item \textbf{Messaging stage.}~
    The principal chooses a (cheap-talk) message profile $\boldsymbol{m}:=\{m_1, \dots, m_n\}$, where $m_i \in \messageSpace$ is the message sent to agent $i$. The message space $\messageSpace$, common for all the agents, is a measurable set that includes $2^{\signalSpace}$ as a subset. 
    \item \textbf{Voting stage.}~
    Each agent $i$ either accepts or rejects the proposal. The project is launched ($L=1$) if at least $k$ agents accept, and not launched ($L=0$) otherwise. 
\end{itemize}

Our paper focuses on comparing public versus private messaging. In terms of the model setup, the only difference between these two communication regimes is the restriction on $\boldsymbol{m}$.
Under public messaging, it is common knowledge that all the agents receive the same message. We let $m_0$ denote this common message, so formally, the principal is subject to the following restriction (both on and off the equilibrium path): $m_i \equiv m_0$, $\forall i$, for some $m_0\in \messageSpace$. 
Under private messaging, we impose no such restriction; this is precisely where the flexibility comes from.

We also introduce a trembling-hand noise in each agent's voting decision. Specifically, each agent is a mechanical type with probability $\epsilon>0$ and a rational type with probability $1-\epsilon$. A rational type votes in his own interest, while a mechanical type randomizes uniformly between acceptance and rejection. The events of being a mechanical type are independent across agents. This noise guarantees that each agent is pivotal with a positive probability. 
In particular, it rules out trivial voting equilibria in which all agents always accept or reject, so no individual is pivotal and the principal cannot affect the outcome through learning and communication.
That being said, we will mainly consider the situation where $\epsilon$ is arbitrarily small. 

Finally, we highlight the principal's \textit{limited commitment} in our setting. Specifically, the principal has ex ante commitment to a test for learning, but no ex post commitment to truthfully reporting her findings, as the communication is via cheap talk. This reflects our main departure from standard information design models.

\subsection{Strategies and Solution Concept}
\label{subsec:strategy-solution-concept}

The principal's strategy includes three components, $(t, \gamma, \sigma)$, corresponding to the three stages in which she needs to make a decision. Her \textit{learning strategy} is the choice of $t\in \testset$. Given $t=(\signalSpace, \pi)$, her \textit{proposal strategy} can be written as $\gamma:\signalSpace\to [0,1]$, where $\gamma(s)$ is the probability of making a proposal upon seeing the signal $s$. Following a proposal, her \textit{messaging strategy} is $\sigma:\signalSpace\rightarrow \Delta(\messageSpace^n)$, where $\sigma(\boldsymbol{m}\mid s)$ is the probability of sending a message profile $\boldsymbol{m}$ to all agents upon seeing the signal $s$. 
Based on $\sigma$, the marginal probability for agent $i$ to receive message $m_i\in\messageSpace$ given signal $s\in\signalSpace$ is denoted by
\begin{equation*}
    \sigma_i(m_i\mid s) := \sum_{\boldsymbol{m}_{-i}\in \messageSpace^{n-1}}\sigma\left((m_i,\boldsymbol{m}_{-i})\mid s\right).
\end{equation*}
Let $\Sigma$ denote the set of all (possibly stochastic) mappings from $\signalSpace$ to $\messageSpace^n$. Under private messaging, the collection of all feasible messaging strategies is exactly $\Sigma$. Under public messaging, due to the restriction of a common message, the collection of all feasible strategies is denoted by $\Sigma_{0} \subsetneq \Sigma$. When no confusion arises, we let $\sigma_0:\signalSpace\rightarrow \Delta(\messageSpace)$ denote the public messaging strategy. 

To describe the agent's strategy, we take the test $t$ as given. From the message $m_i$ and the fact that the principal has proposed, each agent $i$ forms a posterior belief about the state, which we denote by $\psi_i(m_i) \in \Delta(\statespace)$. When he is a rational type, his voting strategy can be written as $\tau_i:\messageSpace\to [0,1]$, where $\tau_i(m_i)$ is his intended acceptance probability upon seeing the message $m_i$. Taking into account the trembling-hand noise, his actual acceptance probability is 
$$
\hat{\tau}_i(m_i) := (1-\epsilon)\cdot \tau_i(m_i)+\frac{\epsilon}{2} \in \left[\frac{\epsilon}{2}, 1-\frac{\epsilon}{2}\right]~.
$$
We let $\boldsymbol{\tau}:=\{\tau_1, \dots, \tau_n\}$ denote the agents' voting strategy profile and $\boldsymbol{\psi}:=\{\psi_1, \dots, \psi_n\}$ their belief profile. 

Given a chosen test $t$, we study the weak Perfect Bayesian equilibrium (wPBE) of the continuation game, represented by $\rho = (\gamma, \sigma, \boldsymbol{\tau}, \boldsymbol{\psi})$ that satisfies the following conditions.\footnote{The off-equilibrium-path scenarios do not play a critical role in the analysis. All the results in this paper remain true if we use less permissive equilibrium concept by restricting the off-equilibrium-path beliefs.} 
\begin{enumerate}
    \item Each player's strategy is sequentially rational given their own belief and the other player's strategy.
    \item For each agent $i$, if there is a positive probability that the project is proposed and the message $m_i$ is sent, his belief $\psi_i(m_i)$ is formed by Bayes' rule. 
    \item For each agent $i$, if there is zero probability that the project is proposed and the message $m_i$ is sent, his belief $\psi_i(m_i)$ is arbitrary.
\end{enumerate}

\section{The Superiority of Public Communication}
\label{sec:general}

In this section, we establish the \emph{superiority of public communication}: for a given test, every equilibrium outcome implementable under private messaging is also implementable under public messaging.
This immediately implies that public messaging achieves a weakly higher payoff for the principal than private messaging.

To begin with,  we consider private messaging and fix a test $t$. For an equilibrium $\rho$ associated with $t$, we let $\messageSpace^{\rho}_{i}$ denote the set of messages that are sent to agent $i$ with positive probability on the equilibrium path.

\begin{lemma}[Restriction by Limited Commitment]
\label{prop:pbe_characterization}
In any equilibrium $\rho$ under private messaging, $\tau_i(m_i)=\tau_i(m_{i}'),\ \forall i, \ \forall m_i, m_i'\in \messageSpace^{\rho}_{i}$. That is, any agent $i$'s probability of accepting the proposal must be \emph{identical} across all messages in $\messageSpace^{\rho}_{i}$. 
\end{lemma}
\begin{proof}
See Appendix \ref{pf:pbe_characterization}. 
\end{proof}

\cref{prop:pbe_characterization} illustrates how the principal's limited commitment restricts equilibrium outcomes. Intuitively, for any agent, if multiple messages lead to his acceptance with different probabilities in equilibrium, the principal, conditional on being willing to propose the project, will always send him the most persuasive message that induces his highest acceptance probability.

As implied by \Cref{prop:pbe_characterization}, under private messaging, the principal essentially only needs one message to each agent upon a proposal. Roughly speaking, this indicates that the principal's message to an agent conveys no payoff-relevant information beyond the fact that the project is proposed, as we formalize in \Cref{prop:private-babbling}. 

\begin{definition}[Outcome]
An outcome is a mapping $\outcome: \statespace \rightarrow [0,1]^2$, where $\outcome_p(\theta)$ denotes the probability that the principal proposes the project when the state is $\theta$, and $\outcome_\ell(\theta)$ denotes the probability that the project is launched conditional on proposal when the state is $\theta$.
\end{definition}

An outcome is defined in this way because the project's proposal and launch are the only two payoff-relevant events in the game, regardless of the communication regime. Given a test $t$, we say that an outcome $\outcome$ is \textit{implementable} under public (private) messaging if there exists an equilibrium under public (private) messaging that induces it. We further let $\outcomeSpace_{\mathsf{pub}}(t)$ and $\outcomeSpace_{\mathsf{pri}}(t)$ denote the set of implementable outcomes given the test $t$ under public and private messaging, respectively. Note that these two sets are both non-empty, since under either communication regime, there always exists a babbling equilibrium as defined below.

\begin{definition}[Babbling Equilibrium]
Under either communication regime, a babbling equilibrium is one where, conditional on the principal proposing the project, her sent message profile is independent of the realized signal, and all the agents make decisions based solely on the fact that the project is proposed.
\end{definition}

\begin{proposition}[Flexibility Is Valueless]
\label{prop:private-babbling}
Given any test $t$, under private messaging, any outcome $\outcome\in\outcomeSpace_{\mathsf{pri}}(t)$ can be implemented by a babbling equilibrium.
\end{proposition}
\begin{proof}
    See Appendix \ref{pf:private-babbling}.
\end{proof}

The intuition behind \Cref{prop:private-babbling} is as follows. According to \Cref{prop:pbe_characterization}, under private messaging, each agent uses the same voting strategy across all messages he may receive on the equilibrium path. Hence, without loss of generality, merging each agent's possible messages on the equilibrium path into a single message will not change the equilibrium outcome. 

\Cref{prop:private-babbling} indicates that, under private messaging, it is without loss of generality to consider \emph{babbling equilibria}, in which the principal's message beyond the proposal itself conveys no payoff-relevant information to individual agents. In other words, communication under private messaging is payoff- and information-irrelevant, as each agent's belief and strategy do not depend on the message, and the principal's payoff also does not depend on the message profile, given that a proposal is made. In this sense, the flexibility afforded by private messaging is essentially ineffective. This observation leads to our main result that public messaging implements (weakly) more outcomes than private messaging, as shown next. 

\begin{theorem}[Public Messaging Implements More]
\label{thm:public-implements-more}
Given any test $t$, any equilibrium outcome that is implementable under private messaging must also be implementable under public messaging using the same test. Formally, $\outcomeSpace_{\mathsf{pub}}(t)\supseteq \outcomeSpace_{\mathsf{pri}}(t)$.
\end{theorem}
\begin{proof}
See Appendix \ref{pf:public-implements-more}. 
\end{proof}

\Cref{thm:public-implements-more} holds because any babbling equilibrium under private messaging continues to be an equilibrium under public messaging. As pointed out in the introduction, the two communication regimes are subject to a flexibility-credibility trade-off. Since the flexibility afforded by private messaging is valueless, public messaging is weakly better, as the credibility that it yields may still be valuable. 

Our analysis so far is based on a given test. If we take into account the optimal choice of test, the weak dominance of public messaging immediately follows from \Cref{thm:public-implements-more}. We define the principal's payoff from a test as her payoff in the principal-preferred equilibrium under this test. The following corollary formalizes the superiority of public communication with the optimally chosen test. 

\begin{corollary}
\label{cor:optimal_test_public_better}
The principal's payoff in the optimal test is weakly higher under public messaging than under private messaging. 
\end{corollary}

\section{When Is Public Communication Strictly Better}
\label{sec:linear}

\Cref{sec:general} establishes that public messaging \emph{weakly} dominates private messaging. A natural next question is: when is public communication \textit{strictly} dominant?

To answer this question in a clear manner, we focus on a canonical environment in which players' payoffs are linear in the underlying state. It is worth noting that the linearity assumption only serves to facilitate our comparison of the two regimes, as it enables us to characterize the principal's optimal test in closed form; it is not essential for the economic insights delivered in this section. 

\Cref{sec:linear} proceeds as follows. After describing the linear environment, \Cref{subsec:public_optimal} and \Cref{subsec:private_optimal} characterize the optimal test under public and private messaging, respectively. Finally, \Cref{subsec:public-strictly-better} compares the two regimes under the optimally chosen test and identifies the necessary and sufficient conditions for public messaging to be strictly dominant.

\paragraph{Linear environment.}
We normalize the principal's payoff from launching the project to
\[
u(\theta)=\theta~.
\]
That is, the states are ordered by the principal's preference. We let each agent $i$'s payoff from launching the project be 
\[
v_i(\theta):=\alpha_i\cdot \theta+\beta_i~,
\]
where $\alpha_i\in\{-1,1\}$ and $\beta_i\in(-1,1)$ are common knowledge.

The parameter $\alpha_i$ determines how agent $i$'s preference aligns with the principal's. When $\alpha_i=1$, agent $i$'s preference is \emph{positively aligned} with the principal's, as they both prefer higher $\theta$; when $\alpha_i=-1$, by contrast, agent $i$'s preference is \emph{negatively aligned} with the principal's. The parameter $\beta_i$ captures agent $i$'s state-independent \emph{bias} about the project. A larger $\beta_i$ indicates that agent $i$ is more predisposed to support the project. Moreover, we define agent $i$'s \emph{acceptance threshold} as 
\[
\delta_i:=-\beta_i/\alpha_i~,
\]
which is the value of $\theta$ to make agent $i$ indifferent between launching the project or not. To avoid uninteresting discussions, we impose a genericity assumption that all agents have distinct acceptance thresholds; that is, $\delta_i\neq \delta_j$ for any $i, j$.

This canonical environment applies broadly to a variety of contexts. As an illustration, consider a legislative committee deciding whether to pass a trade bill (e.g., a package of tariffs and consumer rebates). The bill's impact on the welfare of domestic producers and consumers is ex ante uncertain; we let a larger $\theta$ represent greater benefits for producers, accompanied by larger harms to consumers. The committee is chaired by a \textit{pro-producer} politician, while the remaining members include both \textit{pro-producer} and \textit{pro-consumer} legislators. Pro-producer members may differ in their bias towards the bill, arising from heterogeneity in their constituencies' exposure to global competition. Likewise, pro-consumer members may exhibit varying degrees of opposition, possibly depending on their ties with consumer interests.

\subsection{Optimal Test Under Public Messaging}
\label{subsec:public_optimal}

As mentioned in \Cref{sec:general}, when evaluating a test from the principal's perspective, we consider her payoff from the principal-preferred equilibrium associated with this test. Hence, under public messaging, the principal's problem is as follows. We let $t^*_\mathsf{pub}$ denote a solution to this problem: 
\begin{equation}
\max_{t\in\mathcal{T}}\max_{\outcome\in\outcomeSpace_{\mathsf{pub}}(t)}\int_\Theta \left[u(\theta)\cdot \outcome_l(\theta)-c\right]\cdot \outcome_p(\theta)\dd F(\theta)~,\tag{$\mathcal{P}_\mathsf{pub}$}\label{eq:public-obj}
\end{equation}
where $\outcome_p(\theta)$ is the expected proposal probability at state $\theta$, and $u(\theta)\cdot \outcome_l(\theta)-c$ is the principal's expected net payoff conditional on proposing at state $\theta$.

\begin{definition}[Credible Test] 
Under public messaging, a test $t$ is credible if it admits a \emph{truthful equilibrium}, defined as one in which the principal, conditional on proposing the project, truthfully reports the signal as the public message. 
\end{definition}

\begin{lemma}[Credible Tests Suffice]
\label{lem:public_revelation}
Any outcome $\outcome\in \outcomeSpace_{\mathsf{pub}}(t)$ under some test $t$ is implementable by a credible test $t'$ with the associated truthful equilibrium.
\end{lemma} 
\begin{proof}
    See Appendix~\ref{apx:public-truthful}.
\end{proof}

\Cref{lem:public_revelation} is in the spirit of the revelation principle.\footnote{Similar revelation-principle-type arguments appear independently in other models of information design with limited commitment. See, e.g., \citet{kreutzkamp2024persuasion} and \citet{li2024information}, both of which feature single-agent environments.} It allows us to restrict attention to credible tests, thereby substantially reducing the search space for the optimal test. 

We also need the following lemma that illustrates how the principal's limited commitment restricts equilibrium outcomes under public messaging. Fix a test $t$. For an equilibrium $\rho$ associated with $t$, we let $\signalSpace_{\rho}^{+}$ denote the set of signals under which the principal proposes the project with a positive probability.

\begin{lemma}[Restriction by Limited Commitment]
\label{lem:sameproposeprob}
Given a test $t$ under public messaging, in any equilibrium $\rho$, $\Pr(L=1\mid s, P=1)=\Pr(L=1\mid s', P=1), \forall s, s'\in \signalSpace_{\rho}^{+}$. That is, conditional on the project being proposed, the equilibrium probability that the project is launched is \emph{identical} across all signals in $\signalSpace_{\rho}^{+}$.
\end{lemma} 
\begin{proof}
    See Appendix~\ref{apx:sameproposeprob}.
\end{proof}

The intuition of \Cref{lem:sameproposeprob} is as follows. Given that the principal has proposed the project, launching it must yield an expected payoff higher than the proposal cost $c$, and therefore, she has an incentive to raise its launch probability through strategic communication. As a consequence, if there are two signals in $\signalSpace_{\rho}^+$, say $s$ and $s'$, such that the probability of launching the project under $s$ is higher than $s'$, then the principal, upon receiving the signal $s'$, will profitably deviate to mimic the behavior upon receiving the signal $s$.

\begin{definition}[Bi-pooled Threshold Test]
    A test $t=(\signalSpace,\pi)$ is called a bi-pooled threshold test if (a) $\signalSpace=\{s^-,s^+_1,s^+_2\}$; and (b) there is a threshold $\tilde\theta\in[-1,1]$ and a monotone partition $\Theta_1 \cup \Theta_2 \cup \Theta_3 = [\tilde\theta,1]$ such that 
    \begin{itemize}
        \item $\pi\left(s^-\mid \theta\right)=1$, $\forall \theta\in [-1,\tilde\theta)$;
        \item $\pi\left(s^+_1\mid \theta\right)=1$, $\forall \theta\in \Theta_1\cup\Theta_3$;
        \item $\pi\left(s^+_2\mid \theta\right)=1$, $\forall \theta\in \Theta_2$.
    \end{itemize}
\end{definition}

A bi-pooled threshold test features a threshold $\tau$. It returns the signal $s^-$ if the state lies below the threshold. If the state exceeds the threshold, the test exhibits a bi-pooling structure (see \citet{arieli2023optimal}): among all the states that exceed the threshold, it pools a set of intermediate states (i.e., $\Theta_2$), for which it returns the signal $s_2^+$, and a set of the remaining extreme states (i.e., $\Theta_1\cup\Theta_3$), for which it returns the signal $s_1^+$. 

Note that one degenerate case of a bi-pooled threshold test is a standard threshold test (i.e., one that simply reports whether the state is above a threshold), which, formally, corresponds to $\Theta_2$ being an empty set. Another degenerate case is a double-threshold test, which partitions the state space into three intervals using two thresholds; formally, it corresponds to $\Theta_1$ (or equivalently, $\Theta_3$) being an empty set. 

\begin{proposition}[Optimal Public Test]
\label{prop:publiclinear}
Under public messaging, there exists an optimal test $t^*_{\mathsf{pub}}=(\signalSpace^*_{\mathsf{pub}},\pi^*_{\mathsf{pub}})$ that is a \emph{bi-pooled threshold test} with $\signalSpace^*_{\mathsf{pub}}=\{s^-,s^+_1,s^+_2\}$.
The principal forgoes the project upon receiving $s^-$; upon receiving either $s^+_1$ or $s^+_2$, she proposes the project and truthfully reports the signal received to all agents.
\end{proposition}
\begin{proof}
    See Appendix~\ref{apx:publiclinear}.
\end{proof}

To explain \Cref{prop:publiclinear}, which may appear a bit encrypted, we begin with a sketch of the proof, followed by some intuitive interpretation of the result.

\subsubsection{Proof Sketch of \Cref{prop:publiclinear}}

We sketch the proof using the following example.

\begin{example}
\label{eg1}
Consider an instance:
(a) There are $n=3$ agents. \\
(b) The number of approvals needed to launch the project is $k=1$. \\
(c) The proposal cost is negligible. That is, $c\to 0$.\footnote{Note that we let $c\to 0$ in this example only to simplify the exposition; neither our results nor the proof relies on this assumption.} \\
(d) Agent $1$'s preference is negatively aligned with the principal's, while the other two agents' are positively aligned. Additionally, their acceptance thresholds satisfy $0<\delta_1<\delta_2<\delta_3<1$. 
\end{example}

According to \Cref{lem:public_revelation}, it is without loss of generality that the principal adopts a credible test and truthfully reports the test result as the public message. In a truthful equilibrium, all players form the same posterior belief about the state whenever the project is proposed. Hence, we can determine the principal's posterior payoff as a function of the posterior mean belief about the state, denoted by $\mu$. 

\Cref{subfig:optimal-public-1} illustrates the principal's posterior payoff function. If $\mu\leq 0$, she forgoes the project in the first place, yielding her a posterior payoff of $V(\mu)=0$. If $\mu>0$, she proposes the project, and her posterior payoff is $V(\mu) = \mu \cdot p(\mu)$, where $p(\mu)$ is the probability of launching the project. Specifically, for the four intervals that $\mu$ may fall into, $[0,\delta_1]$, $[\delta_1,\delta_2]$, $[\delta_2,\delta_3]$, and $[\delta_3,1]$, the numbers of approvals from rational agents are $1,0,1,$ and $2$, respectively. Taking into account the trembling-hand error $\epsilon>0$, the corresponding probabilities of launching the project are $p_1, p_0, p_1$, and $p_2$, respectively, with
\begin{equation*}
    p_0:=1-\left(1-\frac{\epsilon}{2}\right)^3~,\quad 
    p_1:=1-\left(1-\frac{\epsilon}{2}\right)^2\left(\frac{\epsilon}{2}\right)~,\quad \mathrm{and}\quad
    p_2:=1-\left(1-\frac{\epsilon}{2}\right)\left(\frac{\epsilon}{2}\right)^2~.
\end{equation*}
Graphically, the posterior payoff function $V(\mu)$ exhibits a piecewise linear structure, as depicted in \Cref{subfig:optimal-public-1}.

\input{figs/EC-2026/optimal-public}

If the principal had full commitment, we could apply the technique developed by \citet{dworczak2019simple} (henceforth, the DM Theorem) to characterize the optimal public test for this example. In our setting, however, the principal has limited commitment, which gives rise to an additional requirement that all signals that induce the principal's proposal must yield the same probability of launching the project, as shown in \Cref{lem:sameproposeprob}. To address this additional requirement, we adopt the following procedure to solve the optimal public test. 

\paragraph{Step 1: Decomposition.} We decompose the principal's problem into three subproblems, each of which restricts the launching probability to be exactly one of $p_0$, $p_1$, or $p_2$. We impose this restriction as a support constraint on the posterior mean. 
For instance, in the subproblem targeting $p_0$, the posterior mean of a proposal signal must lie in the corresponding range $[\delta_1,\delta_2]$

For each subproblem, we construct an auxiliary posterior payoff function, where we set the principal's posterior payoff to be $-\infty$ for the ``banned support'' (i.e., the region of posterior means that should not support a proposal signal). For the subproblem targeting $p_0$, for instance, the auxiliary posterior payoff function $V_0(\mu)$ is identical to the original posterior payoff function when $\mu\in [-1,0] \cup [\delta_1, \delta_2]$, and equals to $-\infty$ when $\mu \in (0,\delta_1)\cup(\delta_2, 1]$. We then apply the DM Theorem to characterize the optimal test of this subproblem,\footnote{Note that the DM Theorem does not apply immediately to the subproblem, since the posterior payoff function does not satisfy the regularity condition specified in \citep{dworczak2019simple}. However, we manage to generalize the DM Theorem to our irregular setting by adopting a finite-penalty regularization, which we will discuss in detail in the proof (see Appendix~\ref{apx:publiclinear}).} which is denoted by $t_0^*$ and will become a candidate solution to the original problem. 

Applying the same procedure, we derive the candidate solutions from the other two subproblems, which we denote by $t_1^*$ and $t_2^*$. An important feature is that every candidate solution is a bi-pooled threshold test (or its degeneration), as illustrated in \Cref{subfig:optimal-public-2,subfig:optimal-public-3,subfig:optimal-public-4}. This structure arises from the fact that the auxiliary posterior payoff function of each subproblem is a ``segmented line'' when $\mu\geq 0$, which we will discuss in detail in \Cref{sec:publiclinearinterpret}. 

\paragraph{Step 2: Consolidation.} We compare the principal's payoffs from the three candidate solutions, $\{t_0^\ast,t_1^\ast,t_2^\ast\}$, and select the maximizer, which is the optimal public test of the original problem. 

Although the proof sketch is based on \Cref{eg1}, the decomposition and consolidation procedure applies more generally and exactly drives \Cref{prop:publiclinear}.

\subsubsection{Interpretation of \Cref{prop:publiclinear}}
\label{sec:publiclinearinterpret}

Having elaborated the proof idea behind \Cref{prop:publiclinear}, we now provide some explanation on why the optimal test must be a bi-pooled threshold test. First, the threshold $\tau$ captures the fact that the principal is willing to launch the project only when the state is sufficiently favorable to herself. 

Second, when the principal proposes the project, she may need more than one proposal signal to secure the agents' acceptance, precisely because of the agents' divergent preferences. As exemplified by \Cref{subfig:optimal-public-3}, where the optimal test consists of two proposal signals,\footnote{Notably, there exist some priors, $t_1^\ast$ is the optimal public test for the original problem in \ref{eg1}.} it harms the principal to merge the two proposal signals: while the signals $s_1^+$ and $s_2^+$ induce the acceptance of Agents 1 and 2, respectively, their merger will cause both agents to vote no to the project. 

Third, although one proposal signal may not be enough, two proposal signals will be sufficient. This result relies on two features of the model in this section. First, the principal's posterior payoff is a collection of several segmented lines, which arise from the linear environment. Second, the principal's limited commitment, which requires that the launching probability must be identical across proposal signals, enables us to examine only one segmented line in each subproblem. As a consequence of this structure, the optimal test induces at most two posterior means from the proposal signals, as formally shown in the proof. 

Fourth, to implement two proposal signals in a test, we can simply adopt the bi-pooling structure for the states above the threshold $\tau$. This is a well-established result from the literature (see \citet{arieli2023optimal}). 

We conclude this subsection with the following observation: the optimal public test does not necessarily come from the subproblem targeting the highest launching probability. Intuitively, the principal faces a trade-off between the level of support (the launching probability) and the size of the admissible support. In particular, a higher launching probability may correspond to a narrower admissible support, which imposes a more restrictive support constraint on the subproblem, making it undesirable. Indeed, the optimal public test must strike the best balance regarding this trade-off.

\subsection{Optimal Test Under Private Messaging}
\label{subsec:private_optimal}

Under private messaging, the principal's problem is as follows. We let $t^*_\mathsf{pri}$ denote a solution to this problem:
\begin{equation}
\max_{t\in\mathcal{T}}\max_{\outcome\in\outcomeSpace_{\mathsf{pri}}(t)}\int_\Theta \left(u(\theta)\cdot \outcome_l(\theta)-c\right)\cdot \outcome_p(\theta)\dd F(\theta)~.\tag{$\mathcal{P}_\mathsf{pri}$}\label{eq:private-obj}
\end{equation}
\begin{lemma}[Binary Tests Suffice]
\label{lem:private-binary-test}
Under private messaging, for any outcome $\outcome$ that is implementable under some test $t=(\signalSpace,\pi)$, there exists another test $t'=(\signalSpace',\pi')$ with a binary signal space $\signalSpace' = \{s^-,s^+\}$ that implements~$\outcome$.
\end{lemma}
\begin{proof}
    See Appendix~\ref{apx:private-binary-test}.
\end{proof}

\Cref{lem:private-binary-test} follows directly from \Cref{prop:private-babbling}. Intuitively, under private messaging, the principal's message conveys no payoff-relevant information to each agent, and therefore, it is without loss of optimality to use a binary test that returns either a proposal signal, which induces the principal to propose the project, or a null signal, which induces her to forgo it. 

It turns out that in the linear environment, an optimal binary test takes the form of a threshold test, as shown in \Cref{prop:privatelinear}.

\begin{definition}[Threshold Test]
A test $t=(\signalSpace,\pi)$ is a threshold test if $\signalSpace=\{s^-,s^+\}$ and there is a threshold $\tilde\theta\in[-1,1]$ such that 
\begin{itemize}
    \item $\pi\left(s^-\mid \theta\right)=1$, $\forall \theta\in [-1,\tilde\theta)$;
    \item $\pi\left(s^+ \mid \theta\right)=1$, $\forall \theta\in [\tilde\theta,1]$.
\end{itemize}
\end{definition}

\begin{proposition}[Optimal Private Test]
\label{prop:privatelinear}
Under private messaging, there exists an optimal test $t^*_{\mathsf{pri}}=(\signalSpace^*_{\mathsf{pri}},\pi^*_{\mathsf{pri}})$ that is a \emph{threshold test} with $\signalSpace^*_{\mathsf{pri}}=\{s^-,s^+\}$.
The principal proposes the project and truthfully reports the signal to all agents upon receiving $s^+$, while forgoing the project upon receiving $s^-$. 
\end{proposition}
\begin{proof}
    See Appendix~\ref{apx:privatelinear}.
\end{proof}

\Cref{prop:privatelinear} suggests that the principal proposes the project if and only if the state exceeds a certain threshold. The intuition is as follows. In the linear environment, the payoff-relevant features of a binary test $t$ are (a) the probability of the signal $s^+$ being realized, which we denote by $p_t$, and (b) the posterior mean belief upon the realization of $s^+$, which we denote by $\mu_t$. Consider, by contradiction, a binary test $t'$ that does not admit the threshold form. We can always construct an alternative threshold test $t''$ such that $\mu_{t''} = \mu_{t'}$ and $p_{t''} > p_{t'}$. This construction strictly improves the principal's payoff, explaining why the optimal private test must take a threshold form.

\subsection{When Is Public Messaging Strictly Better}
\label{subsec:public-strictly-better}

Having characterized the optimal test under each communication regime, we are now ready to address the key question of \Cref{sec:linear}: when does public messaging \textit{strictly} dominate private messaging? As suggested by \Cref{prop:publiclinear} and \Cref{prop:privatelinear}, the optimal public test is a bi-pooled threshold test, whereas the optimal private test is a standard threshold test. Hence, the question boils down to the following: when is the optimal public test not degenerate to a standard threshold test, or in other words, when does it feature two proposal signals? 

To gain some intuition for this question, we investigate \Cref{subfig:optimal-public-3}, where the optimal public test features two proposal signals, $s_1^+$ and $s_2^+$. Roughly speaking, such a test is optimal because there exists an agent who accepts $s_1^+$ but not $s_2^+$, and another agent who does the opposite. This observation leads us to \Cref{thm:strict}, which identifies a necessary and sufficient condition for public messaging to strictly dominate private messaging. 

\begin{definition}[Conflicting Allies]
\label{def:conflicting-allies}
    With linear payoffs, two agents $i\neq j$ are the principal's \emph{conflicting allies} if 
    \[
    \alpha_i=-1,\quad \alpha_j=1,\quad \mathrm{and}\quad 0<\delta_i<\delta_j<1~.
    \]
\end{definition}
To explain \Cref{def:conflicting-allies}, we let the proposal cost be negligible, i.e., $c\to 0$. Two agents are the principal's conflicting allies if they satisfy two conditions. First, they are both the principal's allies, where an ally is defined as an agent whose acceptance region (the set of posterior mean beliefs to make him accept the proposal) overlaps with the principal's acceptance region, which is $(0,1)$ when $c$ is arbitrarily small. This is reflected in the definition as both $\delta_i>0$ and $\delta_j>0$. Second, these two allies are in conflict in the sense that their acceptance regions do not overlap. This is captured in the definition as $\alpha_i=-1$, $\alpha_j=1$, and $\delta_i<\delta_j$. \Cref{fig:conflicting-allies} illustrates the definition.

% =========================================================
% 图 1：冲突盟友 (Conflicting Allies)
% =========================================================
\begin{figure}[t]
    \centering
    \begin{tikzpicture}[scale=1.2] 
        
        % --- 定义视觉坐标 ---
        \coordinate (start) at (-1, 0);     % -1
        \coordinate (zero) at (0, 0);       % 0
        \coordinate (delta_i) at (1.2, 0);  % delta_i = 0.3
        \coordinate (delta_j) at (2.8, 0);  % delta_j = 0.7
        \coordinate (end) at (4, 0);        % 1.0

        % --- 绘制 Agent i 的区间 (红色) ---
        \draw[red, thick, pattern=north east lines, pattern color=red] 
            (start) to[out=60, in=120, looseness=0.6] 
            node[midway, above=2pt, text=red, font=\small] {$A_i=[-1,\delta_i]$} 
            (delta_i) -- cycle;

        % --- 绘制 Agent j 的区间 (蓝色) ---
        \draw[blue, thick, pattern=north west lines, pattern color=blue] 
            (delta_j) to[out=60, in=120, looseness=1.1] 
            node[midway, above=2pt, text=blue, font=\small] {$A_j=[\delta_j,1]$}
            (end) -- cycle;

        % --- 绘制 Gap 花括号 ---
        \draw [decorate, decoration={brace, amplitude=5pt}, thick]
            (1.2, 0.15) -- (2.8, 0.15) 
            node [midway, above=3pt, align=center, font=\small, black] {gap};

        % --- 绘制坐标轴 ---
        \draw[thick, decorate, decoration={zigzag, segment length=5pt, amplitude=2pt}] 
            (-1, 0) -- (-0.4, 0);
        \draw[thick, ->, >=Stealth] (-0.4, 0) -- (4.2, 0) node[right] {$\mu$};

        % --- 刻度与标签 ---
        \foreach \x/\lbl in {-1/-1, 0/0, 4/1} {
            \draw[thick] (\x, 0.1) -- (\x, -0.1);
            \node[anchor=north, inner sep=0pt, yshift=-8pt, font=\small] at (\x, 0) {$\lbl$};
        }
        \foreach \x/\lbl in {1.2/\delta_i, 2.8/\delta_j} {
            \fill[black] (\x, 0) circle (2pt);
            \node[anchor=north, inner sep=0pt, yshift=-8pt, font=\small] at (\x, 0) {$\lbl$};
        }
    \end{tikzpicture}
    \caption{Illustration of two agents $i \neq j$ being the principal's conflicting allies. The red interval indicates Agent $i$'s acceptance region, while the blue interval indicates Agent $j$'s.}
    \label{fig:conflicting-allies}
\end{figure}
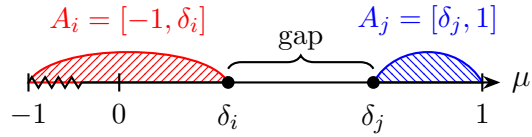

\begin{theorem}[Strict Dominance of Public Messaging]
\label{thm:strict}
In the linear environment, for any trembling-hand error $\epsilon>0$, the following are equivalent.
\begin{enumerate}[label=(\alph*)]
\item There exists a prior distribution $F$ and a cost $c>0$ such that public messaging is strictly better than private messaging. 
\item There exist two agents that are the principal's conflicting allies.
\end{enumerate}
\end{theorem}
\begin{proof}
    See Appendix~\ref{apx:strict}. 
\end{proof}

\subsubsection{Proof Sketch of \Cref{thm:strict}}
\label{sec:thm:strict:proofsketch}

We begin with a definition. Let $N(\mu)$ denote the set of possible numbers of approvals from rational agents given the posterior mean $\mu$.\footnote{The formal definition of $N(\mu)$ shows up in Appendix~\ref{apx:publiclinear}.} Notice that $N(\mu)$ is a singleton if $\mu\neq \delta_i$ for any $i$; if $\mu = \delta_i$ for some $i$, then $N(\mu)$ includes two elements since only agent $i$ is indifferent between acceptance and rejection, given our assumption that all agents have different acceptance thresholds. 

\begin{definition}[Support Gap]
\label{def:supportgap}
An interval $(a, b) \subseteq (0,1)$ is called a support gap if there exists $\Tilde{n}\leq n$ such that (i) $\Tilde{n}\in N(a)$ and $\Tilde{n}\in N(b)$ and (ii) $\max N(x) < \Tilde{n}$ for any $x\in(a, b)$. In other words, while the posterior means, $a$ and $b$, can secure the support of $\Tilde{n}$ rational agents, any posterior mean in the interval $(a, b)$ cannot gain the same level of support. 
\end{definition}

The proof of \Cref{thm:strict} proceeds in two steps. First, we show that the existence of two conflicting allies is equivalent to the existence of a support gap as defined above. Given a support gap $(a, b)$, it is immediate that the two agents who are indifferent in $a$ or $b$ are the principal's conflicting allies. Given two conflicting allies, we can also identify a support gap by construction, whose details are included in the proof. 

Second, we show that the existence of a support gap is equivalent to the strict dominance of public communication under some $(F, c)$. On the one hand, the strict dominance of public communication indicates that an optimal public test needs two proposal signals, which further implies that the merger of these two signals will fall into a support gap, as exemplified by \Cref{subfig:optimal-public-3}. On the other hand, we establish the reverse implication by construction. Given a support gap $(a, b)$, we let $c$ be sufficiently small and the prior distribution $F$ place sufficient weight on $a$ and $b$. In this situation, the two posterior means, $a$ and $b$, represent two groups of agents of equal size, and the principal finds it optimal to secure the support of one and only one of these two groups (selected at random) by using a public test with two proposal signals. Importantly, the principal will not merge the two proposal signals, as doing so would cause her to fall into the support gap.

\subsubsection{Interpretation of \Cref{thm:strict}}

The intuition behind \Cref{thm:strict} is as follows. The presence of two conflicting allies indicates two groups of agents who would support the principal in different states. Roughly speaking, while these two groups may consist of some common members, there are some members in each group who disagree with the other. In such situations, the credibility afforded by public communication is valuable because by publicly siding with one of the two groups, which is selected at random, the principal can secure the trust of the side she chooses.

\Cref{thm:strict} provides a convenient criterion to determine whether public messaging can strictly dominate private messaging. Most strikingly, regardless of the number of agents, it suffices to check whether there exist two agents who are the principal's conflicting allies.

\Cref{thm:strict} is stated for a fixed trembling-hand error $\epsilon>0$. One potential concern is that the strict dominance of public messaging could vanish as $\epsilon\to 0$. \cref{prop:strict-eps-to-zero} addresses this concern by identifying the condition under which the benefit of public messaging persists in the limit.

\begin{proposition}[Strict Dominance When $\epsilon\to 0$]
\label{prop:strict-eps-to-zero}
In the linear environment, the following are equivalent.
\begin{enumerate}[label=(\alph*)]
\item There exists a prior distribution $F$ and a cost $c>0$ such that public messaging is strictly better than private messaging for sufficiently small $\epsilon$;
\item The following set is non-convex:
\begin{equation*}
\Lambda:=\{\theta\in[0,1]: \max N(\theta)\ge k\}~.
\end{equation*}
\end{enumerate}
\end{proposition}
\begin{proof}
    See Appendix~\ref{apx:strict-eps-to-zero}.
\end{proof}

In this proposition, the set $\Lambda$ includes all the states under which the principal and at least $k$ agents are willing to launch the project. \Cref{fig:non-convex-launch-set} illustrates an example where this set is non-convex. Notably, if this set is non-convex, then there must exist two agents who are the principal's conflicting allies; the converse, however, is not necessarily true.
Thus, we provide a stronger condition for the strict dominance of public messaging in \Cref{prop:strict-eps-to-zero}.

Intuitively, as $\epsilon\to 0$, the only first-order determinant of implementation is whether the number of supporters is at least $k$. Conditional on being safely above (or below) the threshold, the distinction between having, say, $k+1$ supporters and $k+10$ supporters affects outcomes only through trembling events of vanishing probability, so it has a negligible effect on payoffs in the limit. Therefore, the mere existence of certain conflicting allies can be washed out by aggregation. For a strict advantage of public over private messaging to persist as $\epsilon\to 0$, the agents' preferences must exhibit a genuine ``support gap'' (see interval $(\delta_i,\delta_j)$ in \Cref{fig:non-convex-launch-set}), i.e., there exists an intermediate rejection region (where fewer than $k$ agents support launching) between two disjoint regions in which (possibly different) coalitions of size at least $k$ support launching. This is exactly what the non-convexity of $\Lambda$ captures.

% =========================================================
% 图 2：非凸启动集 (Non-convex Launch Set)
% =========================================================
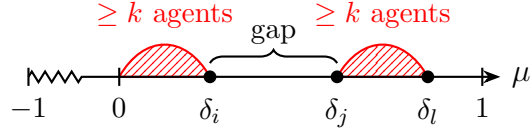
\begin{figure}[t]
    \centering
    \begin{tikzpicture}[scale=1.2]
        
        % --- 定义视觉坐标 ---
        \coordinate (start) at (-1, 0);
        \coordinate (zero) at (0, 0);
        \coordinate (delta_i_new) at (1.0, 0);  % 0.25 * 4 = 1.0
        \coordinate (delta_j_new) at (2.4, 0);  % 0.60 * 4 = 2.4
        \coordinate (delta_l_new) at (3.4, 0);  % 0.85 * 4 = 3.4
        
        % --- 绘制红色弧形区间 ---
        
        % 区间 1: [0, delta_i]
        \draw[red, thick, pattern=north east lines, pattern color=red] 
            (zero) to[out=60, in=120, looseness=1.4] 
            node[midway, above=2pt, text=red, font=\small] {$\ge k$ agents} 
            (delta_i_new) -- cycle;

        % 区间 2: [delta_j, delta_l]
        \draw[red, thick, pattern=north east lines, pattern color=red] 
            (delta_j_new) to[out=60, in=120, looseness=1.4] 
            node[midway, above=2pt, text=red, font=\small] {$\ge k$ agents} 
            (delta_l_new) -- cycle;

        % --- 绘制 Consensus Gap 花括号 ---
        \draw [decorate, decoration={brace, amplitude=5pt}, thick]
            (1.0, 0.15) -- (2.4, 0.15) 
            node [midway, above=3pt, align=center, font=\small, black] {gap};

        % --- 绘制坐标轴 ---
        \draw[thick, decorate, decoration={zigzag, segment length=5pt, amplitude=2pt}] 
            (-1, 0) -- (-0.4, 0);
        \draw[thick, ->, >=Stealth] (-0.4, 0) -- (4.2, 0) node[right] {$\mu$};
        
        % --- 统一刻度 ---
        \foreach \x/\lbl in {-1/-1, 0/0, 4/1} {
            \draw[thick] (\x, 0.1) -- (\x, -0.1);
            \node[anchor=north, inner sep=0pt, yshift=-8pt, font=\small] at (\x, 0) {$\lbl$};
        }
        
        % --- 绘制阈值点 ---
        \foreach \x/\lbl in {1.0/\delta_i, 2.4/\delta_j, 3.4/\delta_l} {
            \fill[black] (\x, 0) circle (2pt);
            \node[anchor=north, inner sep=0pt, yshift=-8pt, font=\small] at (\x, 0) {$\lbl$};
        }
    \end{tikzpicture}
    \caption{An example where the set $\Lambda$ is non-convex. In particular, both the principal and at least $k$ agents support the project if and only if $\theta \in \Lambda=[0,\delta_i]\cup[\delta_j,\delta_l]$.}
    \label{fig:non-convex-launch-set}
\end{figure}

\subsubsection{A Numerical Example of Strict Dominance}

We conclude \Cref{sec:linear} with a numerical example where public messaging is strictly dominant. 

\begin{example}
\label{exp:public-strictly-better}
Consider the following instance:
(a) There are $n=2$ agents. \\
(b) A proposed project is launched as long as $k=1$ agent approves. \\
(c) The state follows a uniform prior distribution, i.e., $F=U[-1,1]$. \\
(d) Two agents' payoff functions are $v_1(\theta)=\theta-0.6$ and $v_2(\theta)=-\theta+0.4$, respectively. \\%These two agents are conflicting allies to the principal.
(e) The proposal cost is negligible, i.e., $c\to 0$. %The trembling-hand error and noise $\epsilon >0$.
\end{example}

By \Cref{prop:publiclinear}, an optimal public test is $t^*_{\mathsf{pub}}=\{\signalSpace^*_{\mathsf{pub}},\pi^*_{\mathsf{pub}}\}$ with $\signalSpace^*_{\mathsf{pub}}=\{s^-,s^+_1,s^+_2\}$ and
\begin{align*}
    &\pi^*_{\mathsf{pub}}\left(s^-\mid \theta\right)=1~,\quad \forall \theta\in[-1,0]~,\\
    &\pi^*_{\mathsf{pub}}\left(s^+_1\mid \theta\right)=1~,\quad \forall \theta\in[0.15,0.65]~,\\
    &\pi^*_{\mathsf{pub}}\left(s^+_2\mid \theta\right)=1~,\quad \forall \theta\in[0,0.15]\cup[0.65,1]~.
\end{align*}
This test brings the principal an expected payoff of
\begin{align*}
    U_{\mathsf{pub}}^*
    &=0.25\cdot 0.4\cdot \left( 1 - \frac{\epsilon}{2} + \frac{\epsilon^2}{4} \right)+0.25\cdot 0.6\cdot \left( 1 - \frac{\epsilon}{2} + \frac{\epsilon^2}{4} \right)\\
    &=0.25\cdot \left( 1 - \frac{\epsilon}{2} + \frac{\epsilon^2}{4} \right)~.
\end{align*}
By \Cref{prop:privatelinear}, an optimal private test is $t^*_{\mathsf{pri}}=\{\signalSpace^*_{\mathsf{pri}},\pi^*_{\mathsf{pri}}\}$ with $\signalSpace^*_{\mathsf{pri}}=\{s^-,s^+\}$ and
\begin{align*}
    &\pi^*_{\mathsf{pri}}\left(s^-\mid \theta\right)=1~,\quad \forall \theta\in[-1,0.2]~,\\
    &\pi^*_{\mathsf{pri}}\left(s^+\mid \theta\right)=1~,\quad \forall \theta\in[0.2,1]~.
\end{align*}
This test brings the principal an expected payoff of
\begin{align*}
    U_{\mathsf{pri}}^*
    =0.4\cdot 0.6\cdot \left( 1 - \frac{\epsilon}{2} + \frac{\epsilon^2}{4} \right)=0.24\cdot \left( 1 - \frac{\epsilon}{2} + \frac{\epsilon^2}{4} \right)<U_{\mathsf{pub}}^*~.
\end{align*}

\section{Summary}
\label{sec:Summary}

This paper studies learning and communication in collective decision problems. A chairperson (the principal) can acquire information about an uncertain project, decide whether to propose it, and then communicate via cheap talk with a committee (a group of agents) that votes under a $k$-majority rule. The paper asks whether, from the principal's perspective, communication should be public or private. The central trade-off is that private communication provides \textit{flexibility} for targeted persuasion, whereas public communication sustains \textit{credibility} by generating common knowledge within the committee. 

The paper delivers two main economic insights. First, the flexibility afforded by private communication is valueless, making public communication weakly dominant. Intuitively, upon proposing the project, the principal has a transparent incentive to maximize each agent's approval probability. As a result, under private communication, her messages convey no additional payoff-relevant information to the agents beyond the proposal itself. 

Second, we identify necessary and sufficient conditions under which the credibility afforded by public communication is valuable, rendering public communication strictly dominant. This occurs when there exist two agents who are the principal's \textit{conflicting allies}: each shares some common interests with the principal, yet they are in conflict with each other. Intuitively, in such situations, public communication allows the principal to publicly side with one of the two conflicting allies, which is selected at random, thereby securing the trust of the side she chooses.

\newpage
\bibliographystyle{apalike}
\bibliography{reference.bib}

@article{li2024information,
  title={Learning and Communication Towards Unanimous Consent},
  author={Li, Yingkai and Xu, Boli},
  journal={arXiv preprint arXiv:2405.18521},
  year={2026}
}

@article{goltsman2011talk,
  title={How to talk to multiple audiences},
  author={Goltsman, Maria and Pavlov, Gregory},
  journal={Games and Economic Behavior},
  volume={72},
  number={1},
  pages={100--122},
  year={2011},
  publisher={Elsevier}
}

@article{arieli2023optimal,
  title={Optimal persuasion via bi-pooling},
  author={Arieli, Itai and Babichenko, Yakov and Smorodinsky, Rann and Yamashita, Takuro},
  journal={Theoretical Economics},
  volume={18},
  number={1},
  pages={15--36},
  year={2023},
  publisher={Wiley Online Library}
}

@article{schnakenberg2017informational,
  title={Informational lobbying and legislative voting},
  author={Schnakenberg, Keith E},
  journal={American Journal of Political Science},
  volume={61},
  number={1},
  pages={129--145},
  year={2017},
  publisher={Wiley Online Library}
}

@article{farrell1989cheap,
  title={Cheap talk with two audiences},
  author={Farrell, Joseph and Gibbons, Robert},
  journal={The American Economic Review},
  volume={79},
  number={5},
  pages={1214--1223},
  year={1989},
  publisher={JSTOR}
}

@article{salcedo2019persuading,
  title={Persuading part of an audience},
  author={Salcedo, Bruno},
  journal={arXiv preprint arXiv:1903.00129},
  year={2019}
}

@article{schnakenberg2015expert,
  title={Expert advice to a voting body},
  author={Schnakenberg, Keith E},
  journal={Journal of Economic Theory},
  volume={160},
  pages={102--113},
  year={2015},
  publisher={Elsevier}
}

@article{bardhi2018modes,
  title={Modes of persuasion toward unanimous consent},
  author={Bardhi, Arjada and Guo, Yingni},
  journal={Theoretical Economics},
  volume={13},
  number={3},
  pages={1111--1149},
  year={2018},
  publisher={Wiley Online Library}
}

@article{nguyen2021bayesian,
  title={Bayesian persuasion with costly messages},
  author={Nguyen, Anh and Tan, Teck Yong},
  journal={Journal of Economic Theory},
  volume={193},
  pages={105212},
  year={2021},
  publisher={Elsevier}
}

@article{lipnowski2022persuasion,
  title={Persuasion via weak institutions},
  author={Lipnowski, Elliot and Ravid, Doron and Shishkin, Denis},
  journal={Journal of Political Economy},
  volume={130},
  number={10},
  pages={2705--2730},
  year={2022},
  publisher={The University of Chicago Press Chicago, IL}
}

@article{dworczak2019simple,
  title={The simple economics of optimal persuasion},
  author={Dworczak, Piotr and Martini, Giorgio},
  journal={Journal of Political Economy},
  volume={127},
  number={5},
  pages={1993--2048},
  year={2019},
  publisher={The University of Chicago Press Chicago, IL}
}

@article{ali2025political,
  title={The Political Economy of Zero-Sum Thinking},
  author={Ali, S Nageeb and Mihm, Maximilian and Siga, Lucas},
  journal={Econometrica},
  volume={93},
  number={1},
  pages={41--70},
  year={2025},
  publisher={Wiley Online Library}
}

@article{kreutzkamp2024persuasion,
  title={Persuasion without ex-post commitment},
  author={Kreutzkamp, Sophie and Lou, Yichuan},
  journal={Journal of Economic Theory},
  pages={106058},
  volume = {228},
  issn = {0022-0531},
  doi = {https://doi.org/10.1016/j.jet.2025.106058},
  year={2025},
  publisher={Elsevier}
}

@article{alonso2016persuading,
  title={Persuading voters},
  author={Alonso, Ricardo and C{\^a}mara, Odilon},
  journal={American Economic Review},
  volume={106},
  number={11},
  pages={3590--3605},
  year={2016},
  publisher={American Economic Association}
}

@article{chan2019pivotal,
  title={Pivotal persuasion},
  author={Chan, Jimmy and Gupta, Seher and Li, Fei and Wang, Yun},
  journal={Journal of Economic theory},
  volume={180},
  pages={178--202},
  year={2019},
  publisher={Elsevier}
}

@article{lin2024credible,
  title={Credible persuasion},
  author={Lin, Xiao and Liu, Ce},
  journal={Journal of Political Economy},
  volume={132},
  number={7},
  pages={2228--2273},
  year={2024},
  publisher={The University of Chicago Press Chicago, IL}
}

@article{lyu2022information,
  title={Information design in cheap talk},
  author={Lyu, Qianjun and Suen, Wing},
  journal={arXiv preprint arXiv:2207.04929},
  year={2024}
}

@article{xu2023,
  title={Optimal Disclosure in Two-Sided Matching},
  author={Xu, Boli},
  journal={Available at SSRN 4003491},
  year={2023}
}

@article{guo2021costly,
  title={Costly miscalibration},
  author={Guo, Yingni and Shmaya, Eran},
  journal={Theoretical Economics},
  volume={16},
  number={2},
  pages={477--506},
  year={2021},
  publisher={Wiley Online Library}
}

@article{lipnowski2020cheap,
  title={Cheap talk with transparent motives},
  author={Lipnowski, Elliot and Ravid, Doron},
  journal={Econometrica},
  volume={88},
  number={4},
  pages={1631--1660},
  year={2020},
  publisher={Wiley Online Library}
}

\newpage
\appendix

\section{Proofs}

\subsection{Proof of \Cref{prop:pbe_characterization}}
\label{pf:pbe_characterization}

Fix a test $t=(\signalSpace,\pi)$. Upon observing a signal $s\in\signalSpace$, the principal forms her posterior belief $\psi(s)\in\Delta(\Theta)$ via Bayes’ rule, which corresponds to an expected payoff of launching the project
\[
\mathbb{E}[u(\theta)\mid s]:=\mathbb{E}_{\theta\sim \psi(s)}\!\big[u(\theta)\big]~.
\]

Suppose the message $m_i$ (respectively, $m_i'$) is sent to agent $i$ with positive probability when the realized signal is $s$ (respectively, $s'$), where $s = s'$ is allowed. 
The fact that she is willing to propose the project under $s$ and $s'$ implies that $\mathbb{E}[u(\theta)\mid s] \geq c >0$ and $\mathbb{E}[u(\theta)\mid s']\geq c>0$.

Let $\boldsymbol{m} \in \text{supp}(\sigma(s))$ (respectively, $\boldsymbol{m}' \in \text{supp}(\sigma(s'))$) be the message profile that specifies the message $m_i$ (respectively, $m_i'$) to agent $i$. 
Let $Y_{-i}$ denote the number of agents other than $i$ who vote yes.
Given the agents' voting strategy $\boldsymbol{\tau}$, for any realized message profile $\boldsymbol{m}$, the equilibrium launch probability conditional on proposal is
\begin{align*}
    \Pr(L=1\mid \boldsymbol{m},P=1)
    =&\Pr(Y_{-i}\ge k\mid \boldsymbol{m}_{-i})\\
    &+\Pr(Y_{-i}= k-1\mid \boldsymbol{m}_{-i})\cdot \hat{\tau}_i(m_i)~,
\end{align*}
where
\begin{equation*}
    \hat{\tau}_i(m_i)=(1-\epsilon)\tau_i(m_i)+\epsilon/2~.
\end{equation*}
Due to the trembling-hand error $\epsilon$, we have
\begin{equation*}
    \Pr(Y_{-i}= k-1\mid \boldsymbol{m}_{-i})>0~.
\end{equation*}

Suppose, by contradiction, $\tau_i(m_i) > \tau_i(m_i')$, then the principal will strictly benefit from deviating from $\boldsymbol{m}'$ to $(m_i, \boldsymbol{m}'_{-i})$ when the realized signal is $s'$, as it strictly increases the probability that the project is launched. Formally, 
\[
\mathbb{E}[u(\theta)\mid s'] \cdot \Pr(L=1\mid (m_i,\boldsymbol{m}'_{-i}), P=1) > \mathbb{E}[u(\theta)\mid s'] \cdot \Pr(L=1\mid \boldsymbol{m}', P=1)~.
\]

\subsection{Proof of \Cref{prop:private-babbling}}
\label{pf:private-babbling}

Fix a test $t=(\signalSpace,\pi)$.
Under private messaging, suppose an outcome $\outcome\in\Omega_{\mathsf{pri}}(t)$ is induced by an equilibrium $\rho$. In this equilibrium, we denote the set of all signals with strictly positive proposal probability by $\signalSpace^{+}_{\rho}\subseteq\signalSpace$. Also, we denote the set of messages that are sent to agent $i$ on the equilibrium path by $\messageSpace_i^\rho\subseteq\messageSpace$. 

\medskip
\noindent
\underline{\emph{Babbling Equilibrium Construction.}}

We construct a new messaging strategy $\sigma'$ by merging all messages in $\messageSpace_i^{\rho}$ into a single message $m^*_i$. 
This message is sent to agent $i$ regardless of the realized signal,
while the messages for other agents remain unchanged; that is,
for any on-path message profile $\boldsymbol{m}$,
\begin{equation*}
    \sigma'((m_i^*,\boldsymbol{m}_{-i})\mid s):=\sigma(\boldsymbol{m}\mid s)~,\quad \forall s\in\signalSpace^{+}_\rho~.
\end{equation*}
By \Cref{prop:pbe_characterization}, agent $i$'s voting strategy is constant across all his on-path messages from $\messageSpace_i^\rho$. Based on this, we construct a new voting strategy and belief for agent $i$. Specifically, define
\begin{equation*}
    \tau_i'(m_i^*):=\tau_i(m_i)~,\quad \forall m_i\in\messageSpace_i^\rho~.
\end{equation*}
Moreover, let agent $i$'s posterior belief after receiving $m_i^*$, which we denote by $\psi_i'(m_i^*)$, be the compound posterior belief from receiving a message from $\messageSpace_i^\rho$ in the original equilibrium $\rho$.

Hence, we construct a new assessment $\rho'$, consisting of the original proposal strategy $\gamma$, the new messaging strategy $\sigma'$, the new voting strategy profile $(\tau_i',\boldsymbol{\tau}_{-i})$, and the new belief system $(\psi_i',\boldsymbol{\psi}_{-i})$.  

\medskip
\noindent
\underline{\emph{Equilibrium and Outcome Implementation Verification.}}

We first consider the voting game induced by $t$, $\gamma$, and $\sigma'$.
Under the modified voting strategy profile, the messages received and strategies of all other agents remain unchanged. 
Hence, agent $i$ faces the same best-response problem as in the original equilibrium $\rho$. Given that he chooses the same voting probability under the belief $\psi_i(m_i)$ across different $m_i\in \messageSpace_i^{\rho}$, it remains optimal for him to do so under the belief $\psi_i'(m_i^*)$. Hence, $(\tau_i', \boldsymbol{\tau}_{-i})$ remains an equilibrium for the continuation game following $t$, $\gamma$, and $\sigma'$.

Given the agents' new voting strategy profile $(\tau_i',\boldsymbol{\tau}_{-i})$, by \Cref{prop:pbe_characterization}, for any two on-path message profiles $\boldsymbol{m}$ and $\boldsymbol{m}'$ in the original equilibrium, we have
\begin{align*}
    \Pr\!\left(L=1\mid (m_i^*,\boldsymbol{m}_{-i}),(\tau_i',\boldsymbol{\tau}_{-i}),P=1\right)
    &= \Pr\!\left(L=1\mid \boldsymbol{m},\boldsymbol{\tau},P=1\right) \\
    &= \Pr\!\left(L=1\mid \boldsymbol{m}',\boldsymbol{\tau},P=1\right) \\
    &= \Pr\!\left(L=1\mid (m_i^*,\boldsymbol{m}'_{-i}),(\tau_i',\boldsymbol{\tau}_{-i}),P=1\right) ~.
\end{align*}
Hence, holding fixed the proposal strategy $\gamma$, the principal's continuation payoff is the same across all on-path message profiles induced by $\sigma'$. It follows that the principal has no profitable deviation from $\sigma'$.
Therefore, we show that $\rho'$ is also an equilibrium given the test $t$. Moreover, we can infer that the constructed equilibrium $\rho'$ implements the same outcome $\outcome$, because it preserves the joint distribution over the state, the proposal decision, and the approval decision.

\vspace{10pt}

By repeating this merging procedure for each agent in an arbitrary order, while preserving the equilibrium and the induced outcome $\outcome$ at every step, we arrive at a strategy where, for any proposal signal from $\signalSpace^{+}_\rho$, the principal proposes the project with strictly positive probability and always sends the same message profile $\boldsymbol{m}^*:= (m_1^*, \dots, m_n^*)$ to all agents. That is, the principal uses a babbling messaging strategy in the new equilibrium, which implements the same outcome $\outcome$.

\subsection{Proof of \Cref{thm:public-implements-more}}
\label{pf:public-implements-more}
Fix a test $t\in \mathcal{T}$ and an outcome $\outcome\in\Omega_{\mathsf{pri}}(t)$ under private messaging. By \cref{prop:private-babbling}, there must be a babbling equilibrium $\rho$ that implements $\outcome$. In this equilibrium, the message space is a singleton that contains the only message $\boldsymbol{m}^*= (m_1^*, \dots, m_n^*)$.

We now construct an assessment $\rho'$ for the game under public messaging. We construct a public message $m_0^*$, keep the principal's proposing strategy unchanged, adjust her messaging strategy by replacing each message $m_i^*$ with $m_0^*$, and let each agent $i$'s belief and voting strategy upon seeing $m_0^*$ be the same as $m_i^*$. This assessment remains to be an equilibrium under public messaging. First, each agent's belief formation and sequential rationality remain unchanged. Second, given that the agents essentially ignore the message sent by the principal beyond the fact that a proposal is made, the principal does not benefit from switching to another message.

\subsection{Proof of \Cref{lem:public_revelation}}
\label{apx:public-truthful}
Consider any test $t=(\signalSpace,\pi)$ and any equilibrium $\rho=(\gamma,\sigma,\tau,\psi)$ under public messaging that implements some outcome $\omega$. We construct another test $t'=(\signalSpace',\pi')$ that admits a truthful equilibrium and implements the same outcome.

\medskip
\noindent
\underline{\emph{Credible Test Construction.}}

First, we construct the new signal space $\signalSpace':=\signalSpace\times \messageSpace$.
For each state $\theta\in\Theta$, define the mapping from state to signal 
\[
\pi'((s,m)\mid \theta):=\pi(s\mid \theta)\sigma(m\mid s)~,
\quad \forall (s,m)\in \signalSpace'~.
\]
Thus, the realized signal records both the original signal $s$ and the public message $m$ that the principal would send in the original equilibrium after observing $s$.

Now consider the following strategy profile under $t'$. After observing $(s,m)\in \signalSpace'$, the principal proposes the project with probability $\gamma(s)$. If she proposes, she truthfully reports $(s,m)$ as the public message. Each agent, upon receiving $(s,m)$, takes the same action as in the original equilibrium after receiving the public message $m$. Formally, for each agent $i$, define
\[
\tau_i'((s,m)):=\tau_i(m)~, \quad \forall (s,m)\in \signalSpace'~.
\]
Let $\psi'$ be the belief system induced by Bayes' rule under $t'$.
Next, we aim to show that the above assessment also forms an equilibrium and preserves the outcome $\outcome$.

\medskip
\noindent
\underline{\emph{Equilibrium and Outcome Implementation Verification.}}

By construction, for every state $\theta$, the joint distribution of $(s,m)$ under $t'$ coincides with the joint distribution of the original signal and the equilibrium public message under $(t,\sigma)$. Hence, after every on-path realization $(s,m)$, the agents' posterior beliefs under $\psi'$ coincide with their posterior beliefs in the original equilibrium after observing the message $m$. Therefore, each agent faces the same continuation problem as in the original equilibrium, so $\tau'$ is sequentially rational.

The principal also faces the same continuation payoff as in the original equilibrium. Indeed, after observing $(s,m)$ under $t'$, choosing proposal probability $\gamma(s)$ and truthfully reporting $(s,m)$ induces exactly the same continuation distribution as observing $s$ and sending $m$ under the original equilibrium. Hence, the principal's strategy is sequentially rational as well.

Therefore, together with $\psi'$, the constructed strategy profile forms a truthful equilibrium under $t'$. Since the joint distribution over the state, the proposal decision, and the launch decision is unchanged, the induced outcome remains $\omega$. Thus, $\omega$ is implementable by a credible test with a truthful equilibrium.

\subsection{Proof of \Cref{lem:sameproposeprob}}
\label{apx:sameproposeprob}

Consider any test $t$ and equilibrium $\rho$ under public messaging.
Suppose, by contradiction, there exists $s, s'\in \signalSpace_{\rho}^{+}$ such that
$$
\Pr(L=1\mid s, P=1) > \Pr(L=1\mid s', P=1)~.
$$
The fact that $s, s'\in \signalSpace_{\rho}^{+}$ implies that $\mathbb{E}[u(\theta)\mid s] \cdot \Pr(L=1\mid s, P=1) -c \geq 0$ and $\mathbb{E}[u(\theta)\mid s'] \cdot \Pr(L=1\mid s', P=1)-c \geq 0$, and thus $\mathbb{E}[u(\theta)\mid s]>0$ and $\mathbb{E}[u(\theta)\mid s']>0$. Hence, if the principal receives the signal $s'$ but mimics the continuation behavior as if she had received signal $s$, her payoff upon a proposal becomes
\begin{equation*}
    \mathbb{E}[u(\theta)\mid s'] \cdot \Pr(L=1\mid s, P=1) -c>\mathbb{E}[u(\theta)\mid s'] \cdot \Pr(L=1\mid s', P=1) -c~.
\end{equation*}

\subsection{Proof of \Cref{prop:publiclinear}}
\label{apx:publiclinear}

\Cref{lem:public_revelation} implies that it is without loss to consider credible tests with the associated truthful equilibrium. In such an equilibrium, once the project is proposed, all players form the same posterior belief. In the linear setting, only the posterior mean is payoff-relevant. It follows that designing an optimal test is equivalent to selecting an optimal mean-preserving contraction (MPC) of the prior. 

However, compared to a standard linear persuasion problem, the designer in our setting is subject to an additional credibility constraint---all signals that induce the principal's proposal must yield the same probability of project launch (\Cref{lem:sameproposeprob}). In other words, all the proposal signals must receive the same number of approvals from rational agents. To address this additional constraint, we first \textit{decompose} the original problem into a family of subproblems, each of which corresponds to a fixed number of approvals from rational agents. After solving all subproblems, we \textit{consolidate} them by comparing all candidate solutions and selecting the best.

\medskip
\noindent \underline{\textbf{Step 1: Decomposition.}}

\smallskip
\noindent \underline{\emph{Step 1.1: Formulate the subproblems.}}

We first determine the number of subproblems. In the linear setting, agent $i$'s approval region is $\mathcal{A}_i=[-1,\delta_i]$ if $\alpha_i=-1$, and $\mathcal{A}_i=[\delta_i,1]$ if $\alpha_i=1$. For any posterior mean $\mu$, let $N(\mu)$ be the set of possible numbers of approvals from rational agents. Formally, 
\begin{equation*}
    N(\mu):=
    \begin{cases}
        \left\{\sum_{i\in[n]}\ind{\mu\in \mathcal{A}_i}\right\}
        & \text{if }\mu\neq \delta_i \text{ for all } i\in[n]~,\\[6pt]
        \left\{\sum_{i\in[n]}\ind{\mu\in \mathcal{A}_i},\ \sum_{i\in[n]}\ind{\mu\in \mathcal{A}_i}-1\right\}
        & \text{otherwise}~.
    \end{cases}
\end{equation*}
Notice that $N(\mu)$ is a singleton when $\mu\neq \delta_i, \forall i\in[n]$, as no agent is indifferent between acceptance and rejection; we also know that $N(\mu)$ contains two elements when $\mu = \delta_i$ for some $i$, as agent $i$ is the only indifferent agent in this situation. 

Denote the number of distinct approval levels that arise in this instance by
\begin{equation*}
    q:=\left|\bigcup_{\mu\in[-1,1]} N(\mu)\right|\le n+1~.
\end{equation*}
Accordingly, the original problem can be decomposed into $q$ subproblems. In each subproblem $j\in[q]$, we let $n_j$ denote the number of approvals from rational agents and $p_j$ denote the corresponding probability of project launch. Formally,
\[p_j := \sum_{l=k}^{n} \sum_{i=0}^{l} \binom{n_j}{i} \binom{n-n_j}{l-i} \left( 1 - \frac{\epsilon}{2} \right)^{n-n_j-l+2i} \left( \frac{\epsilon}{2} \right)^{n_j+l-2i}~.\]
Without loss of generality, we rank the subproblems so that $n_1 < n_2 < \cdots < n_q$; accordingly,  $p_1 < p_2 < \cdots < p_q$. For each $j\in[q]$, we define the \textit{admissible support} (set of posterior means that induce $n_j$ approvals from rational agents) as
\begin{equation*}
    \mathcal{Q}_j:=\left\{\mu\in[-1,1]\mid n_j\in N(\mu)\right\}~.
\end{equation*}
Then, we can formulate the subproblem $j$ as follows:
\begin{align*}
    \max_{G\in\MPC(F)} &\int_{-1}^1 \max\{0,p_j \cdot \mu-c\} \dd G(\mu)\\
   \mathrm{s.t.} \quad &\supp(G)\subseteq [-1,c/p_j]\cup \mathcal{Q}_j 
\end{align*}
Note that, because of Bayesian plausibility, not every subproblem admits a feasible solution; that is, there may not exist an MPC of $F$ that satisfies the support constraint for some subproblems. In Step~2.1 below, we will show that, among the $q$ subproblems, at least one subproblem does. We therefore ignore the infeasible subproblems, as they will not generate the solution to the original problem.

Hereafter, we consider only the feasible subproblems. Suppose subproblem $j$ is feasible. 
To address the support constraint, we introduce a punishment for the non-admissible support, leading to the following problem:  
\begin{align}
    \max_{G\in\MPC(F)} &\int_{-1}^1 V_j(\mu) \dd G(\mu)\tag{$\mathcal{P}_j$}\label{eq:objective-ray-i}
\end{align}
where the principal's posterior payoff function is 
\begin{equation*}
    V_j(\mu):=
    \begin{cases}
        0 & \text{if } \mu\in[-1,c/p_j]~,\\
        p_j\cdot\mu-c & \text{if } \mu\in(c/p_j,1]\cap \mathcal{Q}_j~,\\
        -\infty & \text{if } \mu\in(c/p_j,1]\setminus \mathcal{Q}_j~.
    \end{cases}
\end{equation*}
When $\mu \leq c/p_j$, the principal optimally forgoes the project. When $\mu\in(c/p_j,1]\cap \mathcal{Q}_j$, she proposes the project, which will be launched with probability $p_j$. When $\mu\in(c/p_j,1]\setminus \mathcal{Q}_j$, we assign her an infinite punishment, ensuring that any solution with a finite objective places no mass on such a posterior mean in Problem~\ref{eq:objective-ray-i}. Apparently, subproblem $j$ admits an optimal solution $G^*$ if and only if $G^*$ is optimal to Problem~\ref{eq:objective-ray-i} with a non-negative optimal value.

\medskip
\noindent \underline{\emph{Step 1.2: Transform each feasible subproblem.}}

Problem~\ref{eq:objective-ray-i} is a linear persuasion problem with a continuous state space. The standard methodology for this type of problem is developed in \citet{dworczak2019simple}, which we copy below for convenience. In particular, \Cref{lem:dm-theorem-1} is useful for characterizing an optimal solution, and \Cref{lem:dm-theorem-2} is useful for showing the existence of an optimal solution.

\begin{lemma}[Theorem 1 of \citet{dworczak2019simple}]
\label{lem:dm-theorem-1}
    For any atomless prior $F\in\Delta([0,1])$ and interim utility function $u:[0,1]\to\mathbb{R}$, if there exists a CDF~$G\in\MPC(F)$ and a convex function $\phi:[0,1]\to \mathbb{R}$ such that
    \begin{itemize}
        \item $\phi(x)\ge u(x)$ for all $x\in[0,1]$;
        \item $\supp(G)\subseteq \{x\in[0,1]\mid u(x)=\phi(x)\}$; and
        \item $\int_0^1 \phi(x)\dd G(x)=\int_0^1\phi(x)\dd F(x)$~,
    \end{itemize}
    then $G$ is an optimal solution to $u$ in the space $\MPC(F)$.
\end{lemma}

\begin{definition}[Definition 1 in \citet{dworczak2019simple}]
\label{def:regular}
    A function $u$ is regular if it is upper semi-continuous with at most finitely many one-sided jumps at some $y_1,\dots,y_k\in(0,1)$, and is Lipschitz continuous in each $(y_i,y_{i+1})$ with $y_0=0$ and $y_{k+1}=1$.
\end{definition}

\begin{lemma}[Theorem 2 in \citet{dworczak2019simple}]
\label{lem:dm-theorem-2}
    Suppose that $u$ satisfies the regularity condition in \Cref{def:regular}.
    There exists an optimal solution $G$, and for every optimal solution $G$, there exists a convex and continuous $\phi:[0,1]\to \mathbb{R}$ that satisfies all conditions in \Cref{lem:dm-theorem-1}.
\end{lemma}

The main obstacle that prevents us from applying \Cref{lem:dm-theorem-1,lem:dm-theorem-2} directly is that the posterior payoff function $V_j(\mu)$ is \textit{not} regular in the sense of \Cref{def:regular}, since it is unbounded in the non-admissible support. To deal with this obstacle, we introduce the following regulated problem, where a finite punishment $M\in(0,\infty)$ replaces the infinite punishment on the non-admissible support. 

\begin{equation}
    \max_{G\in\MPC(F)} \int_{-1}^1 V_j'(\mu)\dd G(\mu) \tag{$\mathcal P_j'$}\label{eq:objective-j-B}
\end{equation}
where
\begin{equation*}
    V_j'(\mu):=
    \begin{cases}
        0 & \text{if } \mu\in[-1,c/p_j]~,\\
        p_j\cdot \mu-c & \text{if } \mu\in(c/p_j,1]\cap \mathcal{Q}_j~,\\
        -M & \text{if } \mu\in(c/p_j,1]\setminus \mathcal{Q}_j~.
    \end{cases}
\end{equation*}

When the punishment $M$ is sufficiently large, the optimal solution to Problem~\ref{eq:objective-j-B} converges to the optimal solution to Problem~\ref{eq:objective-ray-i}. Formally,

\begin{lemma}
\label{lem:equivsubproblem}
Suppose subproblem $j$ is feasible. When $M$ is sufficiently large, any optimal solution to Problem~\ref{eq:objective-ray-i} also solves Problem~\ref{eq:objective-j-B}.
\end{lemma}

\begin{proof}
Let $G^*$ denote an optimal solution to Problem~\ref{eq:objective-ray-i} when subproblem $j$ is feasible. Let $G'$ denote an optimal solution to Problem~\ref{eq:objective-j-B}. Note that $G'$ always exists for any $M$ because of \Cref{lem:dm-theorem-2}.

We first show that, for any $M$, the optimal value of Problem~\ref{eq:objective-j-B} is weakly positive.
To see this, since $V_j'=V_j$ on $[-1,c/p_j]\cup \mathcal{Q}_j$ and the same $G^*$ is also feasible for Problem~\ref{eq:objective-j-B}, it holds that
\[
\int_{-1}^1 V_j'(\mu) \dd G'(\mu)\ge\int_{-1}^1 V_j'(\mu) \dd G(\mu)
=
\int_{-1}^1 V_j(\mu) \dd G(\mu)
\ge 0~.
\]

Then, we show that when $M$ is sufficiently large, any optimal solution $G'$ to Problem~\ref{eq:objective-j-B} must satisfy $\supp(G')\subseteq \mathcal{Q}_j\cup[-1,c/p_j]$.
To prove this, write
\begin{equation*}
    \mathcal{Q}_j\cap[c/p_j,1]
    =
    [l_1,r_1]\cup\cdots\cup[l_g,r_g]~,
\end{equation*}
where $c/p_j\le l_1<r_1<\cdots<l_g<r_g\le 1$.
Suppose instead that $\supp(G')\nsubseteq \mathcal{Q}_j\cup[-1,c/p_j]$. By \Cref{lem:dm-theorem-1,lem:dm-theorem-2}, there must then exist some gap index $y\in[g-1]$ and two thresholds $r_y< \tilde{\theta}_1<\tilde{\theta}_2<l_{y+1}$ such that $G'$ is induced by a two-threshold test: mass from $[-1,\tilde{\theta}_1]$ is assigned to $r_y$, mass from $[\tilde{\theta}_1,\tilde{\theta}_2]$ is assigned to some point in $(r_y,l_{y+1})$, and mass from $[\tilde{\theta}_2,1]$ is assigned to $l_{y+1}$. 
The objective under $G'$ is
\begin{align*}
    F(\tilde{\theta}_1)V_j'(r_y)
    +\bigl(1-F(\tilde{\theta}_2)\bigr)V_j'(l_{y+1})+\bigl(F(\tilde{\theta}_2)-F(\tilde{\theta}_1)\bigr)(-M)~.
\end{align*}
For any fixed gap $(r_y,l_{y+1})$, Bayesian plausibility determines at most one such two-threshold test. Hence, whenever such a test exists, we can choose some finite constant $M_y$ large enough that its objective value is strictly negative (hence suboptimal to Problem~\ref{eq:objective-j-B}). Since there are only finitely many gaps, taking $M>\max_y M_y$ over all feasible gaps implies that no optimal solution to Problem~\ref{eq:objective-j-B} can place mass outside the admissible support, i.e., on $(c/p_j,1]\setminus \mathcal{Q}_j$.

Finally, we show that $G^*$ must also be optimal for Problem~\ref{eq:objective-j-B}. Let $G'$ be any optimal solution to Problem~\ref{eq:objective-j-B}. 
Since both $G^*$ and $G'$ are feasible for the original subproblem and $G^*$ is optimal to Problem~\ref{eq:objective-ray-i}, we have
\[
\int_{-1}^1 V_j(\mu)\dd G^*(\mu)
\ge
\int_{-1}^1 V_j(\mu)\dd G'(\mu)~.
\]
Using again that $V_j'=V_j$ on $[-1,c/p_j]\cup \mathcal{Q}_j$, it follows that
\[
\int_{-1}^1 V_j'(\mu)\dd G^*(\mu)
\ge
\int_{-1}^1 V_j'(\mu)\dd G'(\mu)~.
\]
Therefore, $G^*$ is also an optimal solution to Problem~\ref{eq:objective-j-B}.
This proves the lemma.
\end{proof}

\medskip
\noindent \underline{\emph{Step 1.3: Solve each feasible subproblem.}}

Because of \Cref{lem:equivsubproblem}, we can switch gears to Problem~\ref{eq:objective-j-B} with sufficiently large $M$. Since the function $V_j'(\mu)$ satisfies the regularity conditions, \Cref{lem:dm-theorem-2} implies the existence of an optimal solution $G^*_j$ for the subproblem $j$, which is fully characterized by \Cref{lem:dm-theorem-1}. In particular, there exists a convex function $\phi_j:[-1,1]\to\mathbb{R}$ such that:
\begin{enumerate}
    \item $\phi_j(\mu) \ge V_j'(\mu)$ for all $\mu \in [-1,1]$;
    \item $\supp(G^*_j) \subseteq \{\mu\in[-1,1] \mid \phi_j(\mu) = V_j'(\mu)\}$;
    \item $\int_{-1}^1 \phi_j(x) \dd  G^*_j(x) = \int_{-1}^1 \phi_j(x) \dd  F(x)$.
\end{enumerate}

From exhausting all possible forms of $\phi_j$, we can infer that $G^*_j$ can be implemented by a bi-pooled threshold test $t^*_j=(\signalSpace^*_j,\pi^*_j)$.
Specifically, the test is parameterized by a threshold $\tilde\theta_j\in[-1,1]$ and uses at most three signals $\signalSpace^*_j=\{s^-,s^+_1,s^+_2\}$. 
There exists a monotone partition of $[\tilde\theta_j,1]$ into three disjoint sets $\Theta_1$, $\Theta_2$ and $\Theta_3$ such that
\begin{align*}
    &\pi^*_j(s^-\mid \theta)=1~,\quad\mathrm{for\ all\ } \theta\in[-1,\tilde\theta_j]~,\\
    &\pi^*_j(s^+_1\mid \theta)=1~,\quad\mathrm{for\ all\ } \theta\in\Theta_1\cup\Theta_3~,\\
    &\pi^*_j(s^+_2\mid \theta)=1~,\quad\mathrm{for\ all\ } \theta\in\Theta_2~.
\end{align*}
By the DM theorem, the convex function $\phi_j$ must be piecewise linear and can have at most one kink. The kink corresponds to the threshold $\tilde{\theta}_j$. The tangency, or overlap, between the second linear segment of $\phi_j$ and $V_j'$ then determines whether we need one or two signals in the region above the threshold.

\medskip
\noindent
\underline{\textbf{Step 2: Consolidation.}}

The \emph{consolidation} step collects the optimal solutions to all subproblems and selects the one most preferred by the principal.

\medskip
\noindent \underline{\emph{Step 2.1: There exists a feasible subproblem in any instance.}}

Let $\mu_0:=\E[\theta]$ denote the prior mean. If $\mu_0\le c$, then every subproblem admits a feasible solution, since pooling all states into a single non-proposal signal is feasible and yields the principal a payoff of zero. If $\mu_0\in(c,\max_{j\in[q]} c/p_j]$, then any subproblem $j$ satisfying $c/p_j\ge \mu_0$ admits a feasible solution for the same reason. Finally, if $\mu_0\in(\max_{j\in[q]} c/p_j,1]$, then pooling all states into a single proposal signal is feasible and yields the principal a strictly positive payoff for some subproblem. In particular, there exists some $j\in[q]$ such that $\mu_0\in \mathcal{Q}_j$.

\medskip
\noindent\underline{\emph{Step 2.2: Select the best among these feasible subproblems.}}

We have shown that every feasible subproblem $j$ admits an optimal solution $G_j^*$. By Step~1, for sufficiently large $M$, $G_j^*$ is also optimal for Problem~\ref{eq:objective-j-B}. Moreover, \Cref{lem:dm-theorem-1} fully characterizes $G_j^*$, and in particular implies that $G_j^*$ is induced by a bi-pooled threshold test. Hence, we end up with at least one candidate solution from the subproblems. We may therefore select, among all these candidate solutions, the one that maximizes the principal's payoff. This yields the optimal public test.

\subsection{Proof of \Cref{lem:private-binary-test}}
\label{apx:private-binary-test}

Fix any test $t$. By \Cref{prop:private-babbling}, for any outcome $\outcome\in\Omega_{\mathsf{pri}}(t)$ implementable under private messaging, there exists a babbling equilibrium $\rho=(\gamma,\sigma,\boldsymbol{\tau},\boldsymbol{\psi})$ that implements $\outcome$. In this equilibrium, after any signal $s\in \signalSpace_\rho^+$ and conditional on proposal, the principal always sends the same message profile $\boldsymbol{m}^*\in\messageSpace^n$. Moreover, by \Cref{lem:sameproposeprob}, all signals from $\signalSpace_\rho^+$ induce the same launch probability, denoted by $p^*>0$.

\medskip
\noindent
\underline{\emph{Binary Test Construction.}}

We now construct a binary test $t'=(\signalSpace',\pi')$ and an equilibrium $\rho'=(\gamma',\sigma',\boldsymbol{\tau}',\boldsymbol{\psi}')$
that implements the same outcome $\outcome$. Let $\signalSpace':=\{s^-,s^+\}$ and define
\begin{equation*}
    \pi'(s^+\mid \theta):=\sum_{s\in\signalSpace}\pi(s\mid \theta)\gamma(s)
    \quad\text{and}\quad
    \pi'(s^-\mid \theta):=1-\pi'(s^+\mid \theta)~,
    \quad \forall \theta\in\Theta~.
\end{equation*}
Thus, under $t'$, signal $s^+$ is realized exactly with the probability that the principal proposes under the original equilibrium.
Given $t'$, define the proposal strategy by
\begin{equation*}
    \gamma'(s^+):=1
    \quad\text{and}\quad
    \gamma'(s^-):=0~.
\end{equation*}
Conditional on proposal, let the principal always send the same message profile $\boldsymbol{m}^*$:
\begin{equation*}
    \sigma'(\boldsymbol{m}^*\mid s^+):=1~.
\end{equation*}
Finally, we let $\boldsymbol{\tau}' = \boldsymbol{\tau}$ and $\boldsymbol{\psi}'=\boldsymbol{\psi}$. 

\medskip
\noindent
\underline{\emph{Equilibrium and Outcome Implementation Verification.}}

We next show that $\rho'$ is an equilibrium under $t'$. First, we consider the voting stage. Whenever the principal proposes, the agents receive the same message profile $\boldsymbol{m}^*$. By construction of $\pi'$, for every state $\theta$, $\pi'(s^+\mid \theta)$ is exactly the probability that the project is proposed under the original equilibrium. Hence, the posterior belief induced by a proposal under $t'$ coincides with the posterior belief induced by a proposal under the original equilibrium. Therefore, the agents' beliefs $\boldsymbol{\psi}'$ remain consistent and their strategies $\boldsymbol{\tau}'$ remain sequentially rational.

Second, we consider the communication stage. Given that the agents essentially ignore the principal's message in a babbling equilibrium, $\sigma'$ remains optimal. 

Third, we consider the proposal stage. If the principal observes $s^+$ under $t'$, then proposing and sending $\boldsymbol{m}^*$ yields launch probability $p^*$ (since the voting behaviors of the agents remain unchanged). Let $\mu^+$ denote her posterior belief after observing $s^+$. By Bayes' rule, $\mu^+$ is a convex combination of the posterior beliefs induced by those original signals $s$ with $\gamma(s)>0$.
Hence, $\E_{\theta\sim \mu^+}[u(\theta)]\cdot p^*-c$ is also a convex combination of the $\E[u(\theta)\mid s]\cdot p^*-c$ over signals $s$ with $\gamma(s)>0$. By sequential rationality of $\gamma$ in the original equilibrium, each of these quantities is weakly nonnegative. Therefore,
\[
\E_{\theta\sim \mu^+}[u(\theta)]\cdot p^*-c\ge 0~,
\]
and thus proposing is optimal after $s^+$.

Similarly, if $s^-$ is realized with positive probability, let $\mu^-$ denote the posterior belief after observing $s^-$. Then $\mu^-$ is a convex combination of the posterior beliefs induced by those original signals $s$ with $\gamma(s)<1$.
If the principal deviates and proposes after observing $s^-$, then, since the equilibrium is babbling, sending $\boldsymbol{m}^*$ again induces launch probability $p^*$. Hence her deviation payoff is $\E_{\theta\sim \mu^-}[u(\theta)]\cdot p^*-c$, which is a convex combination of the quantities $\E[u(\theta)\mid s]\cdot p^*-c$ over signals $s$ with $\gamma(s)<1$. By sequential rationality of $\gamma$ in the original equilibrium, each of these quantities is weakly nonpositive. Therefore,
\[
\E_{\theta\sim \mu^-}[u(\theta)]\cdot p^*-c\le 0~,
\]
and thus forgoing is optimal after $s^-$. 

Finally, we verify that $\rho'$ implements the same outcome $\outcome$. For every state $\theta$, the proposal probability $\pi'(s^+\mid\theta)$ is exactly the proposal probability under the original equilibrium. The argument thus follows.

\subsection{Proof of \Cref{prop:privatelinear}}
\label{apx:privatelinear}

By \Cref{lem:private-binary-test}, under private messaging, it is without loss to restrict attention to binary tests with the associated truthful equilibrium. Thus, it suffices to show that any binary test can be replaced by a threshold test that weakly improves the principal's payoff.

Consider any binary test $t=(\signalSpace,\pi)$ with $\signalSpace=\{s^-,s^+\}$, together with a truthful equilibrium in which the principal forgoes after $s^-$ and proposes after $s^+$. Let
\[
\mu^+:=\E[\theta\mid s^+]
\quad\text{and}\quad
\mu^-:=\E[\theta\mid s^-]
\]
denote the posterior mean values induced by the two signals. Let $p(\mu^+)$ denote the launch probability conditional on proposal after $s^+$; under the linear environment, this probability depends only on $\mu^+$. The principal's optimal proposal decision implies
\[
\mu^+\cdot p(\mu^+)-c\ge 0
\quad\text{and}\quad
\mu^-\cdot p(\mu^+)-c\le 0~.
\]
Hence $\mu^+ \ge \mu^-$. If $\mu^+=\mu^-$, then the binary test is uninformative, and the conclusion is immediate. We therefore focus on the case where $\mu^+>\mu^-$. By Bayesian plausibility, we have $\mu^-<\E[\theta]<\mu^+$. 

Next, we construct an alternative threshold test $t'$ that preserves the posterior mean induced by $s^+$ and weakly raises its ex ante probability. The same construction appears as Lemma~2 in \citet{li2024information}. Define 
\[
\phi(r):=\E[\theta\mid \theta\ge r]
=
\frac{\int_r^1 \theta \dd F(\theta)}{1-F(r)}~,
\quad \forall r\in[-1,1)~.
\]
Because $F$ is atomless, $\phi$ is continuous and weakly increasing. Moreover,
\[
\phi(-1)=\E[\theta]<\mu^+
\quad\text{and}\quad
\lim\limits_{r\rightarrow 1}\phi(r)=1~.
\]
Hence, there exists $r^*\in[-1, 1]$ such that $\phi(r^*)=\mu^+$. Define a threshold test $t'=(\signalSpace,\pi')$ with the mapping
\[
\pi'(s^+\mid \theta):=\ind{\theta\ge r^*}
\quad \mathrm{and} \quad
\pi'(s^-\mid \theta):=1-\pi'(s^+\mid \theta)~,\quad \forall \theta\in[-1,1]~.
\]
This threshold test preserves the posterior mean induced by $s^+$, as $\E_{t'}[\theta\mid s^+] = \mu^+$. Also, it weakly increases the ex ante probability of $s^+$. By the construction of a threshold test, any non-threshold binary test whose ex ante probability of $s^+$ is $\Pr(s^+\mid t')$ must have $s^+$-induced posterior mean strictly lower than $\mu^+$, and therefore, there cannot be any non-threshold binary test with $s^+$-induced posterior mean being $\mu^+$ that is associated with an ex ante probability of $s^+$ being weakly greater than $\Pr(s^+\mid t')$. In other words, 
we have $\Pr(s^+\mid t')\ge \Pr(s^+\mid t)$.

Such a construction will weakly increase the principal's payoff. Specifically, when we restrict attention to binary tests in a linear environment, the principal's payoff can be pinned down by the posterior mean $\mu^+$ induced by $s^+$, together with the ex ante probability of $s^+$. Fixing $\mu^+$ while increasing the ex ante probability of $s^+$ weakly raises the principal's payoff.

\subsection{Proof of \Cref{thm:strict}}
\label{apx:strict}

We prove in two steps, as sketched in \Cref{sec:thm:strict:proofsketch}. The two lemmas below correspond to these two steps, respectively. Without loss of generality, let $-1=\delta_0<\delta_1<\delta_2<\dots<\delta_n<\delta_{n+1}=1$.

\begin{lemma}
\label{lem:exist-consensus-gap}
There exists a support gap if and only if there exist two conflicting allies.
\end{lemma}

\begin{proof}
We prove the two directions in turn.

\medskip
\noindent
\underline{\emph{Proof of Sufficiency.}}

Suppose $(a,b)$ is a support gap. By its definition, there exists $\tilde{n}\le n$ such that
\[
\tilde{n}\in N(a)~, \quad \tilde{n}\in N(b)~, \quad \text{and}\quad \max N(x)<\tilde{n}~, \quad \forall x\in(a,b)~.
\]
We claim that both $a$ and $b$ must be acceptance thresholds of some agents. If $a\neq \delta_i$ for all $i\in[n]$, then $N(\mu)$ is singleton and locally constant around $a$, which contradicts $\max N(x)<\tilde{n}$ for all $x>a$ sufficiently close to $a$. Hence $a=\delta_i$ for some unique agent $i$. By the same argument, $b=\delta_j$ for some unique agent $j$.

Next, we show $\alpha_i=-1$ and $\alpha_j=1$. Since $\tilde{n}\in N(a)$ but $\max N(x)<\tilde{n}$ for all $x\in(a,b)$, the number of supporters must drop to the right of $a$. This is possible only if agent $i$ rejects posterior means slightly above $a$, that is, only if $\alpha_i=-1$. Similarly, $\alpha_j=1$. Thus, we have
\[
0<\delta_i=a<b=\delta_j<1~.
\]
Therefore, agents $i$ and $j$ are the principal's conflicting allies.

\medskip
\noindent
\underline{\emph{Proof of Necessity.}}

Suppose there exist two conflicting allies $i\neq j$ with
\[
\alpha_i=-1~,\quad \alpha_j=1~\quad \text{and}\quad 0<\delta_i<\delta_j<1~.
\]
Consider all acceptance thresholds in the interval $[\delta_i,\delta_j]$, ordered increasingly. Since the corresponding sequence of signs starts with $-1$ and ends with $1$, there must exist two adjacent agents $d$ and $d+1$ with $\alpha_d=-1$
and $\alpha_{d+1}=1$.
Let
\[
a:=\delta_d
\quad\text{and}\quad
b:=\delta_{d+1}~.
\]
$N(x)$ is a singleton and constant for all $x\in(a,b)$. Let $N(x)=\{\bar n\}$. At $a$, agent $d$ is indifferent, so he may either accept or reject; hence $\bar n+1\in N(a)$.
In the same manner, $\bar n+1\in N(b)$.
Besides, we have $\max N(x)=\bar n<\bar n+1$ for any $x\in(a,b)$.
Thus, $(a,b)$ forms a support gap.
\end{proof}

The following lemma establishes the equivalence between the existence of a support gap and the strict dominance of public messaging.

\begin{lemma}%[Equivalence between support gap and the Strict Superiority of Public Messaging]
\label{lem:consensus-gap-strict}
There exists a prior $F$ and a cost $c>0$ such that public messaging strictly outperforms private messaging if and only if there exists a support gap.
\end{lemma}

\begin{proof}
We prove the two directions separately.

\medskip
\noindent
\underline{\emph{Proof of Sufficiency.}}

Suppose there exists a support gap $(a,b)$. 
Note that, for the prior $F$ constructed below, there exists a threshold $C>0$ such that the same argument applies for every $c\in(0,C)$. 
We may therefore take $c$ to be arbitrarily small, and work directly with the limiting case $c\to 0$.

First, if we ignore the continuous assumption of the prior, we can directly construct such a binary distribution only supported at $\theta=a$ and $\theta=b$:
\begin{equation*}
    F^*(\theta):=
    \begin{cases}
        0\quad &\mathrm{if\ }\theta\in[-1,a)~,\\
        0.5\quad &\mathrm{if\ }\theta\in[a,b)~,\\
        1\quad &\mathrm{if\ }\theta\in[b,1]~.
    \end{cases}
\end{equation*}

Under private messaging, \Cref{prop:privatelinear} implies that an optimal test can be taken to be a threshold test.
Moreover, the cutoff $\tilde{\theta}$ must satisfy $\tilde{\theta}<a$; otherwise, the test would generate two proposal signals, contradicting the binary constraint.
Hence, the optimal private test pools the two masses into a single proposal posterior mean $(a+b)/2$ and yields the principal an expected payoff of
$$
U^*_{\mathsf{pri}}=\frac{a+b}{2}\cdot N\left(\frac{a+b}{2}\right)~.
$$

Under public messaging, by \Cref{prop:publiclinear}, there exists an optimal test that takes the form of a bi-pooled threshold test.
The optimal public test implements a bi-pooling design on $[0,1]$ (which coincides with full revelation under the binary prior) and yields the principal an expected payoff of
\begin{align*}
    U^*_{\mathsf{pub}}
    &=\frac{1}{2}\cdot a \cdot N(a)+\frac{1}{2}\cdot b \cdot N(b)\\
    &=\frac{a}{2}\cdot N(a)+\frac{b}{2}\cdot N(b)> \frac{a+b}{2}\cdot N\left(\frac{a+b}{2}\right)=U^*_{\mathsf{pri}}~,
\end{align*}
where the inequality holds since $(a,b)$ is a support gap.

Next, we impose the continuity requirement on the prior. 
We construct a sequence $\{F^m\}_{m\in\mathbb{Z}^+}$ of continuous, piecewise-linear CDFs on $[-1,1]$ that converges weakly to the binary distribution $F^*$. 
Intuitively, for sufficiently large $m$, each $F^m$ is bimodal, with almost all of its mass concentrated in small neighborhoods of $a$ and $b$, and with the two neighborhoods receiving approximately equal probability mass.

Specifically, for each $m\in\mathbb{Z}^+$, let $\varepsilon_m=1/m$ and for any state $\theta\in[-1,1]$
\begin{equation*}
F^m(\theta):=
\begin{cases}
0 & \mathrm{if\ }\theta\in[-1,\,a-\varepsilon_m)~,\\
\dfrac{\theta-(a-\varepsilon_m)}{4\varepsilon_m} 
& \mathrm{if\ }\theta\in[a-\varepsilon_m,\,a+\varepsilon_m)~,\\
1/2 & \mathrm{if\ }\theta\in[a+\varepsilon_m,\,b-\varepsilon_m)~,\\
1/2+\dfrac{\theta-(b-\varepsilon_m)}{4\varepsilon_m}
& \mathrm{if\ }\theta\in[b-\varepsilon_m,\,b+\varepsilon_m)~,\\
1 & \mathrm{if\ }\theta\in[b+\varepsilon_m,\,1]~.
\end{cases}
\end{equation*}
Obviously, by construction, each $F^m$ is continuous and piecewise linear. 
Let $\varphi$ be any bounded continuous function on $[-1,1]$. Then
\begin{align*}
\int_{-1}^1 \varphi(\theta)\dd  F^m(\theta)
&=\int_{a-\varepsilon_m}^{a+\varepsilon_m}\varphi(\theta)\dd F^m(\theta)
  +\int_{b-\varepsilon_m}^{b+\varepsilon_m}\varphi(\theta)\dd F^m(\theta)\\
&=\frac12\cdot\frac{1}{2\varepsilon_m}\int_{a-\varepsilon_m}^{a+\varepsilon_m}\varphi(\theta)\dd \theta
 +\frac12\cdot\frac{1}{2\varepsilon_m}\int_{b-\varepsilon_m}^{b+\varepsilon_m}\varphi(\theta)\dd \theta~.
\end{align*}
Since $\varphi$ is uniformly continuous on the compact set $[-1,1]$, it holds that
\[
\frac{1}{2\varepsilon_m}\int_{a-\varepsilon_m}^{a+\varepsilon_m}\varphi(\theta)\dd \theta\to \varphi(a)
\quad\mathrm{and}\quad
\frac{1}{2\varepsilon_m}\int_{b-\varepsilon_m}^{b+\varepsilon_m}\varphi(\theta)\dd \theta\to \varphi(b)~.
\]
Therefore,
\[
\int_{-1}^1 \varphi(\theta)\dd F^m(\theta)\to \frac12\,\varphi(a)+\frac12\,\varphi(b)
=\int_{-1}^1 \varphi(\theta)\dd F^*(\theta)~,
\]
which proves the weak convergence of $F^m\Rightarrow F^*$ as $m\to \infty$.

For sufficiently large $m$, under the prior $F^m$, any optimal private test pools all posterior mass in $[0,1]$ into a single proposal signal. Consequently, the induced proposal posterior converges to $(a+b)/2$, and the principal's payoff satisfies
\begin{equation*}
    \lim_{m\to\infty} U_{\mathsf{pri}}^m
    =
    \frac{a+b}{2}\cdot N\left(\frac{a+b}{2}\right)~.
\end{equation*}

By contrast, for sufficiently large $m$, an optimal public test adopts a bi-pooling structure over $[0,1]$. The corresponding payoff, therefore, satisfies
\begin{align*}
    \lim_{m\to\infty} U_{\mathsf{pub}}^m
    &=
    \frac{a}{2}\cdot N(a)+\frac{b}{2}\cdot N(b)>
    \frac{a+b}{2}\cdot N\left(\frac{a+b}{2}\right)
    =
    \lim_{m\to\infty} U_{\mathsf{pri}}^m~,
\end{align*}
where the strict inequality follows from the fact that $(a,b)$ is a support gap.

\medskip
\noindent
\underline{\emph{Proof of Necessity.}}

Suppose that, for some prior $F$ and cost $c>0$, the optimal public test strictly outperforms the optimal private test. 
By the proof of \Cref{prop:publiclinear}, any optimal public test must solve one of the subproblems in our decomposition. 
Moreover, \Cref{lem:dm-theorem-1,lem:dm-theorem-2} imply that the optimizer of that subproblem takes a bi-pooled threshold structure above some cutoff $\tau$: it generates two proposal signals $s_1^+$ and $s_2^+$ with posterior means denoted by $a$ and $b$ that satisfy $0<a<b<1$. 

Given such an optimum, any pooling of the proposal signals $s_1^+$ and $s_2^+$ cannot improve the principal's payoff. In particular, if these two signals are merged into a single proposal signal, the resulting posterior mean is the mixture $\bar{\mu}\in(a,b)$, while the ex ante proposal probability remains unchanged. 
Since the original bi-pooling scheme is optimal for the subproblem, and therefore also for the overall public problem, such a pooling deviation cannot yield a weak improvement.

It follows that the mixture posterior $\bar{\mu}$ must induce a strictly lower launch probability than the common launch probability induced by posterior means in neighborhoods of $a$ and $b$. Hence, $(a,b)$ constitutes a support gap.
\end{proof}

\subsection{Proof of \Cref{prop:strict-eps-to-zero}}
\label{apx:strict-eps-to-zero}

For any trembling-hand error $\epsilon>0$, there are at most $n+1$ distinct launch probabilities that can be induced in equilibrium. As $\epsilon\to 0$, these probabilities converge to the two extreme values, $0$ and $1$, because each agent's realized vote is then almost surely determined by his intended action. In this limit, the posterior means that induce a launch probability arbitrarily close to $1$ are summarized by the launch set $\Lambda$. 

The set $\Lambda$ is non-convex if and only if there exists a support gap $(a,b)\subseteq(0,1)$. 
Equivalently, there exist posterior means on the two sides of $(a,b)$ that each generate at least $k$ supporters and hence induce launch with probability one, whereas any posterior mean within $(a,b)$ induces launch with probability close to zero. 
In that case, using the same construction as in the proof of \Cref{thm:strict}, we can construct a bimodal prior that places most of its mass in small neighborhoods of $a$ and $b$. 
Under such a prior, the optimal public test strictly outperforms the optimal private test.

\end{document}